\documentclass[runningheads,10pt]{llncs}
\usepackage{graphicx}
\usepackage{amsmath}
\usepackage{mathrsfs}
\usepackage{amssymb}
\usepackage{mdframed}
\usepackage{tcolorbox}

\usepackage{amssymb}
\usepackage{epsfig}

\usepackage[normalem]{ulem}

\usepackage{todonotes}
\usepackage{soul,xcolor}

\usepackage{algpseudocode}
\usepackage{xspace}

\usepackage{comment}
\excludecomment{supress}

\usepackage{mathrsfs}  
\usepackage{tabularx}
\usepackage{multirow}
\usepackage{booktabs}
\usepackage{array}

\usepackage[inline]{enumitem}

\usepackage{nicefrac}

\usepackage{cleveref}

\algrenewcommand\algorithmicrequire{\textbf{Input:}}

\newcolumntype{C}[1]{>{\centering\let\newline\\\arraybackslash}m{#1}}
\usepackage{csquotes}

\newtheorem{prop}{Proposition}
\newtheorem{clm}{Claim}

\usepackage{amsopn}
\newtheorem{observation}{Observation}

\newtheorem{reduction}{Reduction Rule}

\newcommand{\OO}{{\mathcal O}}

\newcommand{\nph}{\textsf{NP}-hard\xspace}

\newcommand{\yes}{{\sc Yes}\xspace}

\newcommand{\fpt}{{\sf FPT}\xspace}
\newtheorem{branching rule}{Branching Rule}
\newtheorem{reduction rule}{Reduction Rule}

\newcommand{\FPT}{\textsf{\textup{FPT}}\xspace}

\newcommand{\WOH}{\textsf{\textup{W[1]-hard}}\xspace}

\newcommand{\W}{\textsf{\textup{W}}}

\newcommand{\Co}[1]{\ensuremath{\mathcal{#1}}\xspace}
\newcommand{\X}[1]{\ensuremath{\mathscr{#1}}\xspace}

\newcommand{\sse}{\subseteq}

\newcommand{\wtkp}{{\sf 2D-KP}\xspace}

\newcommand{\fit}{\text{\sc Size-Fitting}\xspace}
\newcommand{\phase}{\text{Phase}\xspace}

\newcommand{\smp}{{\sc Stable Marriage}\xspace}
\newcommand{\ie}{i.e.\xspace}
\newcommand{\wrt}{{with respect to}\xspace}
\renewcommand{\wp}{{with probability}\xspace}

\newcommand{\asm}{{\sc ASM}\xspace}

\newcommand{\amsfullnew}{{\sc Local Search-ASM}\xspace}
\newcommand{\ams}{{\sc LS-ASM}\xspace}

\newcommand{\mcq}{{\sc Multicolored Clique}\xspace}
\newcommand{\mcqsmall}{{\sc MCQ}\xspace}

\newcommand{\wtknapsack}{{\sc Two-Dimensional Knapsack}\xspace}

\newcommand{\defproblem}[3]{
  \vspace{1mm}
\noindent\fbox{
  \begin{minipage}{.96\textwidth}
  \begin{tabular*}{\textwidth}{@{\extracolsep{\fill}}lr} #1 \\ \end{tabular*}
  {\bf{Input:}} #2  \\
  {\bf{Question:}} #3
  \end{minipage}
  }
  \vspace{1mm}
}

\newcommand{\Ma}[1]{\textcolor{magenta}{#1}}

\newcommand{\hide}[1]{}

\newcommand{\il}[1]{\todo[inline, color={yellow!60}]{\small{#1}}}
\newcommand{\ma}[1]{\todo[color={green!70}]{\tiny{#1}}}

\definecolor{rvwvcq}{rgb}{0.08235294117647059,0.396078431372549,0.7529411764705882}
\definecolor{dtsfsf}{rgb}{0.8274509803921568,0.1843137254901961,0.1843137254901961}
\definecolor{mygreen}{rgb}{0.1803921568627451,0.49019607843137253,0.19607843137254902}

\newcommand{\pref}[1]{\ensuremath{\langle #1 \rangle}}

\newcommand{\roof}[1]{\ensuremath{\lceil #1 \rceil}}

\newcommand{\basevertex}[2]{\ensuremath{#1_{#2}}}

\newcommand{\specialvertex}[3]{\ensuremath{#1_{#2}^{#3}}}

\begin{document}

\title{On the (Parameterized) Complexity of Almost Stable Marriage 
}
%
%
\author{Sushmita Gupta\inst{1}\and
Pallavi Jain\inst{2}\and
Sanjukta Roy\inst{3}\and
Saket Saurabh\inst{3,4}\and
Meirav Zehavi\inst{2}}
\authorrunning{S. Gupta et al.}
%
\institute{National Institute for Science Education and Research, HBNI, Bhubaneswar, India.\\
\email{sushmitagupta@niser.ac.in} \and
Ben-Gurion University of the Negev, Be’er Sheva, Israel\\
\email{pallavi@post.bgu.ac.il, meiravze@bgu.ac.il}\and
The Institute of Mathematical Sciences, HBNI, Chennai, India.\\
\email{sanjukta@imsc.res.in, saket@imsc.res.in}\and
University of Bergen, Norway
}
\maketitle              

\begin{abstract}
In the {\sc Stable Marriage} problem, when the preference lists are complete, all agents of the smaller side can be matched. However, this need not be true when preference lists are incomplete. In most real-life situations, where agents participate in the matching market voluntarily and submit their preferences, it is natural to assume that each agent wants to be matched to someone in his/her preference list as opposed to being unmatched. In light of the Rural Hospital Theorem, we have to relax the ``no blocking pair'' condition for stable matchings in order to match more agents. In this paper, we study the question of matching more agents with fewest possible blocking edges. In particular, we find a matching whose size exceeds that of stable matching in the graph by at least $t$ and has at most $k$ blocking edges. 
We study this question in the realm of parameterized complexity with respect to several natural parameters, $k,t,d$, where $d$ is the maximum length of a preference list. Unfortunately, the problem remains intractable even for the combined parameter $k+t+d$. Thus, we extend our study to the local search variant of this problem, in which we search for a matching that not only fulfills each of the above conditions but is ``closest'', in terms of its symmetric difference to the given stable matching, and obtain an \FPT algorithm.  
\end{abstract}


\section{Introduction}\label{sec:intro}


Matching various entities to available resources is of great practical importance, exemplified in matching college applicants to college seats, medical residents to hospitals, preschoolers to kindergartens, unemployed workers to jobs, organ donors to recipients, and so on. It is noteworthy that in the applications mentioned above, it is not enough to merely match an entity to any of the available resources. It is imperative, in fact, mission-critical, to create matches that fulfil some predefined notions of compatibility, suitability, acceptability, and so on. Gale and Shapley introduced the fundamental theoretical framework to study such two-sided matching markets in the 1960s. They envisioned a matching outcome as a {\it marriage} between the members of the two sides, and a desirable outcome representing a {\it stable marriage}. The algorithm proffered  by them has since attained wide-scale recognition as the Gale-Shapley stable marriage/matching algorithm \cite{GS62}. Stable marriage (or stable matching, in general) is one of the acceptability criteria for matching in which an unmatched pair of agent should not prefer each other over their matched partner. 
\par
Of the many characteristic features of the two-sided matching markets, there are certain aspects that stand out and are supported by both theoretical and empirical evidence--particularly notable is the curious aspect that for a given market with strict preferences on both sides,\footnote{In most real-life applications, it is unreasonable if not unrealistic to expect each of the agents to rank all the agents on the other side. That is, the graph $G$ is highly unlikely to be complete.} no matter what the stable matching outcome is, the specific {\em number} of resources matched on either side always remains the same. This fact encapsulated by The Rural Hospital's Theorem \cite{RHT1,RHT2} states that  no matter what stable matching algorithm is deployed, the exact {\em set} (rather than only the number) of resources that are matched on either side is the same. In other words, {\it there is a trade-off between size and stability such that any increase in size must be paid for by sacrificing stability.} Indeed, it is not hard to find instances in which as much as half of the available resources are unmatched in every stable matching. Such gross underutilization of critical and potentially expensive resources has not gone unaddressed by researchers. In light of the Rural Hospital Theorem, many variations have been considered, some important  ones being:  enforcing lower and upper capacities, forcing some matches, forbidding some matches, relaxing the notion of stability, and finally foregoing stability altogether in favor of size~\cite{BFWM10j,BMM10j,CM16j,IrvingManlove09b,MMT17j,Tomoeda18j}. 
\par
%
%
We formalize the  trade-off mentioned above between size and stability in terms of the {\sc Almost Stable Marriage} problem.  The classical {\sc Stable Marriage} problem takes as an instance, a bipartite graph $G = (A \cup B, E)$, where $A$ and $B$ denote the set of vertices representing the agents on the two sides and $E$ denotes the set of edges representing acceptable matches between vertices on different sides, and a preference list of every vertex in $G$ over its neighbors.  Thus, the length of the preference list of a vertex is same as its degree in the graph.
A {\it matching} is defined as a subset of the set of edges $E$ such that no vertex appears in more than one edge in the matching. An edge in a matching represents a match such that the endpoints of a matching edge are said to be the {\it matching partners} of each other, and an unmatched vertex is deemed to be self-matched.  A matching $\mu$ is said to be {\it stable} in $G$ if there does not exist a {\it blocking edge} \wrt $\mu$, defined to be an edge $e\in E\setminus \mu$ whose endpoints rank each other higher (in their respective preference lists) than their matching partners in $\mu$.\footnote{Every candidate is assumed to prefer being matched to any of its neighbors to being self-matched.} The goal of the \smp problem is to find a stable matching. We define the {\sc Almost Stable Marriage} problem as follows.

\defproblem{{\sc Almost Stable Marriage} ({\sc ASM})}{A bipartite graph $G = (A\cup B, E)$, a set $\mathcal{L}$ containing the preference list of each vertex, and non-negative integers $k$ and $t$.}{Does there exist a matching whose size is at least $t$ more than the size of a stable matching in $G$ such that the matching has at most $k$ blocking edges?} 

In {\sc ASM}, we are happy with a matching that is larger than a stable matching but may contain some blocking edges. The above problem quantifies these two variables: $t$ denotes the minimum increase in size, and $k$ denotes the maximum number of blocking edges we may tolerate. 
\par
We note that Bir\'o et al.~\cite{BMM10j} considered the problem of finding, among all matchings of the maximum size, one that has the fewest blocking edges, and showed the {\sf NP}-hardness of the problem even when the length of every preference list is at most three.  Since one can find a maximum matching and a stable matching in the given graph in polynomial time~\cite{micali1980v,GS62}, their {\sf NP}-hardness result implies  {\sf NP}-hardness for {\sc ASM} even when the length of every preference list is at most three by setting $t=\texttt{size of a maximum matching}- \texttt{size of a stable matching}$. 
We study the parameterized complexity of {\sc ASM}  with respect to parameters, $k$ and $t$, 
which is not implied by their reduction. Our first result exhibits a strong guarantee of intractability.
\begin{theorem}
\label{thm:asm}
{\sc ASM} is \W$[1]$-hard \wrt $k+t$, even when the  maximum degree is at most four. 
\end{theorem}

%
We prove Theorem~\ref{thm:asm}, by showing a polynomial-time many-to-one parameter preserving reduction from the \mcq (\mcqsmall) problem on the regular graphs to {\sc ASM}. In a {\em regular graph}, the degree of every vertex is the same. In the \mcq problem on regular graphs, given a regular graph $G=(V,E)$ and a partition of $V(G)$ into $k$ parts, say $V_1,\ldots,V_k$; the goal is to decide the existence of a set $X\subseteq V(G)$ such that $\lvert X\cap V_i\rvert =1$, for all $i\in [k]$, and $G[X]$ induces a clique, that is, there is an edge between every pair of vertices in $G[X]$. \mcqsmall is known to be \W$[1]$-hard on regular graphs~\cite{cai2008parameterized}. 
\par

In light of the intractability result in Theorem~\ref{thm:asm}, we are hard pressed to recalibrate our expectations of what is algorithmically feasible in an efficient manner. Therefore, we consider local search approach for this problem, in which, instead of finding any matching whose size is at least $t$ larger than the size of stable matching, we also want this matching to be ``closest'', in terms of its symmetric difference, to a stable matching. Such framework of local search has also been studied for other variants of the {\sc Stable Marriage} problem by Marx and Schlotter~\cite{DBLP:journals/disopt/MarxS11,DBLP:journals/algorithmica/MarxS10}. It has also been studied for several other optimization problems~\cite{FFLRSV12j,KBP03,KrokhinMarx08,MarxPCN,Marx08,MarxSchlotter10,MarxSchlotter11,Szeider11}.  This question is formally defined as follows.

\defproblem{\amsfullnew(\ams)}{A bipartite graph $G = (A\cup B, E)$, a set $\mathcal{L}$ containing the preference list of every vertex, a stable matching $\mu$, and non-negative integers $k$, $q$, and $t$.}{Does there exist a matching $\eta$ of size at least $|\mu|+t$ with at most $k$ blocking edges such that the symmetric difference between $\mu$ and $\eta$ is at most $q$?}

\newcommand{\Cal}[1]{\ensuremath{\mathcal{#1}}\xspace}


\smallskip 
 Unsurprisingly perhaps, the existence of a stable matching in the proximity of which we wish to find a solution, does not readily mitigate the computational hardness of the problem, as evidenced by Theorem~\ref{thm:whard-k-and-t}, which is implied by the construction of an instance in the proof of Theorem~\ref{thm:asm} itself.

\begin{theorem}\label{thm:whard-k-and-t}
\ams is \W$[1]$-hard \wrt $k+t$, even when maximum degree is at most four. 
\end{theorem}
In our quest for a parameterization that makes the problem tractable, we investigate \ams\ \wrt $k+q+t$. 




\begin{theorem}\label{thm:smd w-hard}\label{thm:whard-k-q-t}
\ams is \W$[1]$-hard \wrt $k+q+t$. 
\end{theorem}

To prove Theorem~\ref{thm:smd w-hard},  we again give a polynomial-time many-to-one parameter preverving reduction from the \mcqsmall problem to \ams. We wish to point out here that in the instance which was constructed to prove Theorem~\ref{thm:asm}, $q$ is not a function of $k$. Thus, we mimic the idea of gadget construction in the proof of Theorem~\ref{thm:asm} and reduces $q$ to a function of $k$. However, in this effort, degree of the graph increases. Therefore, the result in Theorem~\ref{thm:smd w-hard} does not hold for constant degree graph or even when the degree is a function of $k$. This tradeoff between $q$ and the degree of the graph in the construction of instances to prove intractability results is not a coincidence as implied by our next result.

\begin{supress}
\il{***********Don't need this now**********}

\begin{corollary}
\ams is \WOH \wrt $k+q$. 
\end{corollary}

\il{*******************}
\end{supress}



\begin{supress}

The reduction that yields the above theorem, can be executed in polynomial time, and thus also yields an NP-hardness result for \ams, as stated below.

\begin{theorem}\label{thm:nph}
\ams is \nph even when every vertex has at most four neighbors. 
\end{theorem}

That is, it is highly unlikely to have an algorithm for \ams with running time $\OO(n^{f(d)})$ (and hence also $f(d) n^{\OO(1)}$) for any computable function $f$ of $d$, where $d$ denotes the maximum degree of the graph. Phrased in the language of Parameterized Complexity, \ams is para-\nph\ when parameterized by $d$.

\end{supress}


\begin{theorem}\label{th:fptAlgo}
There exists an algorithm that given an instance of  \ams, solves the instance in $2^{\OO(q \log d)+o(dq)}n^{\OO(1)}$ time, where $n$ is the number of vertices in the given graph, and $d$ is the maximum degree of the given graph.
\end{theorem}

To prove Theorem~\ref{th:fptAlgo}, we use the technique of {\em random separation} based on color coding, in which the underlying idea is to highlight the solution that we are looking for with good probability. Suppose that $\eta$ is a hypothetical solution to the given instance of \ams. Note that to find the matching $\eta$, it is enough to find the edges that are in the symmetric difference of $\mu$ and $\eta$ ($\mu \triangle \eta$). Thus, using the technique of random separation, we wish to highlight the edges in $\mu \triangle \eta$. We achieve this goal 
%
%
using two layers of randomization. The first one separates vertices that appear in  $\mu \triangle \eta$, denoted by the set $V(\mu \triangle \eta)$, from its neighbors, by independently  coloring vertices $1$ or $2$. Let the vertices appearing in $V(\mu \triangle \eta)$ be colored $1$ and its neighbors that are not in $V(\mu \triangle \eta)$ be colored $2$. Observe that the matching partner of the vertices which are not in $V(\mu \triangle \eta)$ is same in both $\mu$ and $\eta$. 
%
Therefore, we search for a solution locally in vertices that are colored $1$. Let $G_1$ be the graph induced on the vertices that are colored $1$. At this stage we use a second layer of randomization on edges of $G_1$, and independently color each edge with $1$ or $2$. This separates edges that belong to $\mu \triangle \eta$ (say colored $1$) from those that do not belong to $\mu \triangle \eta$. Now for each component of $G_1$,  
we look at the edges that have been colored $1$, and compute the number of blocking edges, the increase in size and increase in the symmetric difference, if we modify using the  $\mu$-alternating paths/cycle that are present in this component. This leads to an instance of the \wtknapsack (\wtkp) problem, which we solve in polynomial time using a known pseudo-polynomial time algorithm for the \wtkp problem~\cite{DBLP:books/daglib/0010031}.  We  derandomize this algorithm using the 
notion of an $n$-$p$-$q$-{\em lopsided universal} family~\cite{FominLPS16}.  

\smallskip

\noindent{\bf Related Work:} 
We present here some variants of the {\sc Stable Marriage} problem which are closely related to our model. For some other variants of the problem, we refer the reader to ~\cite{chen2019matchings,david2013algorithmics,gusfield1989stable,knuth1976marriages}.
\par
In the past, the notion of ``almost stability'' is defined for the {\sc Stable Roommate} problem~\cite{abraham2005almost}. In the {\sc Stable Roommate} problem, the goal is to find a stable matching in an arbitrary graph. As opposed to {\sc Stable Marriage}, in which the graphs is a bipartite graph, an instance of {\sc Stable Roommate} might not admit a stable matching. Therefore, the notion of almost stability is defined for the {\sc Stable Roommate} problem, in which the goal is to find a matching with a minimum number of blocking edges. This problem is known as  the {\sc Almost Stable Roommate} problem. Abraham et al.~\cite{abraham2005almost} proved that the {\sc Almost Stable Roommate} problem is \nph. Biro et al.~\cite{biro2012almost} proved that the problem remains \nph even for constant-sized preference lists and studied it in the realm of approximation algorithms. Chen et al.~\cite{DBLP:conf/icalp/ChenHSY18} studied this problem in the realm  of parameterized complexity and showed that the problem is \W[1]-hard with respect to the number of blocking edges even when the maximum length of every preference list is five.
\par
Later in $2010$, Bir\'o et al.~\cite{BMM10j} considered the problem of finding, among all matchings of the maximum size, one that has the fewest blocking edges, in a bipartite graph and showed that the problem is \nph and not approximable within $n^{1-\epsilon}$, for any $\epsilon >0$ unless {\sf P}={\sf NP}. 
\par
The problem of finding the maximum sized stable matching in the presence of ties and incomplete preference lists, max{\sc SMTI}, has striking resemblance with \asm. In max{\sc SMTI}, the decision of resolving each tie comes down to deciding who should be at the top of each of tied lists, mirrors the choice we have to make in \asm in rematching the vertices who will be part of a blocking edge in the new matching. Despite this similarity, the \W$[1]$-hardness result presented in \cite[Theorem 2]{MarxSchlotter11} does not yield the hardness result of  \asm  and \ams
as the reduction is not likely to be parameteric in terms of $k+t$ and $k+t+q$, or have the degree bounded by a constant. 
\section{Preliminaries}
{\em Sets.} We denote the set of natural numbers $\{1,\ldots,\ell\}$ by $[\ell]$.  For two sets $X$ and $Y$, we use notation $X\triangle Y$ to denote the symmetric difference between $X$ and $Y$. We denote the union of two disjoint sets $X$ and $Y$ as $X \uplus Y$. For any ordered set $X$, and an appropriately defined value $t$, $X(t)$ denotes the $t^{\text{th}}$ element of the set $X$. Conversely, suppose that $x$ is $t^{\text{th}}$ element of the set $X$, then $\sigma(x,X)=t$. \par 
\noindent{\em Graphs.} Let $G$ be an undirected graph. We denote the vertex set and the edge set of $G$ by $V(G)$ and $E(G)$ respectively.  We denote an edge between $u$ and $v$ as $uv$, and refer $u$ and $v$ as the {\it endpoints} of the edge $uv$.  
The {\em neighborhood} of a vertex $v$, denoted by $N_G(v)$, is the set of all vertices adjacent to it. Analogously, the {\em (open) neighborhood} of a subset $S\sse V$, denoted by $N_{G}(S)$, is the set of vertices outside $S$ that are adjacent to some vertex in $S$. Formally, $N_G(S)=\cup_{v\in S}N_G(v)$. The degree of a vertex $v$ is the graph $G$ is the number of vertices in $N_G(v)$. The maximum degree of a graph is the maximum degree of its vertices, that is, for the graph $G$, the maximum degree is $\max_{v\in V(G)}|N_G(v)|$. A graph is called a {\em regular graph} if the degree of all the vertices in the graph is the same. For regular graph, we call the maximum degree of the graph as the degree of the graph.  A \emph{component} of $G$ is a maximal subgraph in which any two vertices are connected by a path. For a component $C$, $N_G(C)=N_G(V(C))$. The subscript in the notation may be omitted if the graph under consideration is clear from the context. \par
 In the preference list of a vertex $u$, if $v$ appears before $w$, then we say that $u$ prefers $v$ more than $w$, and denote it as $v \succ_u w$. We call an edge in the graph as {\em static edge} if its endpoints prefer each other over any other vertex in the graph.
For a matching $\mu$, $V(\mu)=\{u,v\colon uv\in \mu\}$. If an edge $uv\in \mu$, then $\mu(u)=v$ and $\mu(v)=u$. A vertex is called {\em saturated} in a matching $\mu$, if it is an endpoint of one of the edges in the matching $\mu$, otherwise it is an {\em unsaturated} vertex in $\mu$. If $u$ is an unsaturated vertex in a matching $\mu$, then we say $\mu(u)=\emptyset$. For a matching $\mu$ in $G$, a $\mu$-alternating path(cycle) is a path(cycle) that starts with an unsaturated vertex and whose edges alternates between matching edges of $\mu$ and non-matching edges. A $\mu$-augmenting path is a $\mu$-alternating path that starts and ends at an unmatched vertex in $\mu$. \par
 Unless specified, we will be using all general graph terminologies from the book of Diestel~\cite{diestel-book}. For parameterized complexity related definitions, we refer the reader to~\cite{ParamAlgorithms15b,DowneyFbook13,niedermeier06b}.

 \begin{prop}\label{cl:bp vertices}
Let $\mu$ and $\mu'$ denote two matchings in $G$ such that $\mu$ is stable and $\mu'$ is not.  
Then, for each blocking edge with respect to $\mu'$ we know that at least one of the endpoints has different matching partners in $\mu$ and $\mu'$. 
 \end{prop}

\begin{proof}
Let $uv$ be a blocking edge with respect to $\mu'$. Towards the contrary, suppose that $\mu'(u)=\mu(u)$ and $\mu'(v)=\mu(v)$. 
Since $uv$ is a blocking edge with respect to $\mu'$, we have that $v \succ_u \mu'(u)$, and $u \succ_v \mu'(v)$. Therefore, $v \succ_u \mu(u)$, and $u \succ_v \mu(v)$. Hence, $uv$ is also a blocking edge with respect to $\mu$, a contradiction to that $\mu$ is a stable matching in $G$.
\end{proof}

\section{W[1]-hardness of {\sc ASM}}\label{sec:construction}
%

We give a polynomial-time parameter preserving many-to-one reduction from the \WOH problem \mcq (\mcqsmall) (\cite{cai2008parameterized}) on regular graphs. \hide{in which we are given a regular graph $G=(V,E)$ and a partition of $V(G)$ into $k$ parts, $V_1,\ldots,V_k$, and the objective is to decide if there exists a subset $S\subseteq V(G)$ such that $| S\cap V_i| =1$, for each $i\in [k]$, and the induced subgraph $G[S]$ is a clique. }

It will be necessary for us to assume that certain sets are ordered. This ordering uniquely defines the $t^{th}$ element of the set (for an appropriately defined value of $t$), and thereby enables us to refer to the $t^{th}$ element of the set unambiguously. We assume that sets $V_{i}$ (for each $i\in [k]$) and $E_{ij}$ (for each $\{i,j\} \sse [k], i<j$) have a canonical order, and thus for an appropriately defined value $t$, $V_{i}(t)$ ($E_{ij}(t)$) and $\sigma(V_{i}, v)$ ($\sigma(E_{ij}, e)$) are uniquely defined. For ease of exposition, for any vertex $v\in V(G')$ we will refer to its set of neighbors, as an ordered set. In such a situation we will denote $N(v)=\pref{ \cdot, \cdot}$. 

Given an instance $\Co{I}=(G,(V_1,\ldots,V_k))$ of \mcqsmall, where $G$ is a regular graph whose degree is denoted by $r$, we will next describe the construction of an instance $\Co{J}=(G',\Co{L},k',t)$ of \asm. 

  \par
 {\bf Construction.} We begin by introducing some notations. For any $\{i, j\} \sse [k]$, such that $i<j$, we use $E_{ij}$ to denote the set of edges between sets $V_i$ and $V_j$. For each $i\in [k]$, we may assume that $| V_i| =n=2^p$, and for each $\{i,j\} \subseteq [k]$, we may assume that $\lvert E_{ij} \rvert =m=2^{p'}$, for some positive integers $p$ and $p'$ greater than one.\footnote{Let $p$ be the smallest positive integer greater than one such that $n<2^p$, add $2^p-n$ isolated vertices in $V_i$. Similarly, let $p'$ be the smallest positive integer greater than one such that $m<2^{p'}$, add  $2^{p'}-m$ isolated edges (an edge whose endpoints are of degree exactly one) to $E_{ij}$. Note that if $(G,(V_1,\ldots,V_k))$ was a \WOH instance of \mcqsmall earlier, then so even now.}.


 
 
For each  $j\in [\log_2 (\nicefrac{n}{2})]$, let $\beta_j= \nicefrac{n}{2^j}$, and $\gamma_j=\nicefrac{n}{2^{j+1}}$. For each  $j\in [\log_2(\nicefrac{m}{2})] $, let $\rho_j=\nicefrac{m}{2^j}$, and $\tau_j= \nicefrac{m}{2^{j+1}}$. 
Next, we are ready to describe the construction of the graph $G'$.


\begin{itemize}
\item[ ]~~~
\noindent\hspace{-0.7cm}{\bf Base vertices:}
\item For each vertex $u \in V(G)$, we have $2r+2$ vertices in $G'$, denoted by $\{\basevertex{u}{i} : i\in [2r+2]\}$, connected via a path: $(\basevertex{u}{1},\ldots,\basevertex{u}{2r+2})$. 

 \item For each edge $e\in E(G)$, we have vertices  $\basevertex{e}{}$ and $\basevertex{\tilde{e}}{}$ in $G'$ that are neighbors. 
 
 
\hide{For any $u\in V(G)$, we define $\X{E}_{u}=\{e \in V(G') | e \in E_{u}\}$ and assume that the set $\X{E}_{u}$ has the same canonical order as $E_{u}$.} 
 

\item  For each $h\in [r]$, \basevertex{u}{2h+1} is a neighbor of the vertex \basevertex{e}{}, where $e=\sigma(E_{u}, h)$. 

\smallskip
\noindent\hspace{-0.3cm}{\bf Special vertices.} For each $i\in [k]$, we define a set of special vertices as follows. 
\item For each $\ell \in [\beta_1]$, we add vertices \specialvertex{p}{\ell}{i} and \specialvertex{\tilde{p}}{\ell}{i} to $V(G')$. Let $u$ and $v$ denote the $2\ell-1^{st}$ and the $2\ell^{th}$ vertices in $V_{i}$, respectively. Then, the vertex 
\specialvertex{p}{\ell}{i} is a neighbor of vertices \basevertex{u}{1} and \basevertex{v}{1}; and the vertex \specialvertex{\tilde{p}}{\ell}{i} is a neighbor of vertices \basevertex{u}{2r+2} and \basevertex{v}{2r+2} in $G'$.
 
\item For each $j\in [\log_2 (\nicefrac{n}{2})]$ and $\ell \in [\beta_j]$, we add vertices $\specialvertex{a}{j,\ell}{i}$ and $\specialvertex{\tilde{a}}{j,\ell}{i}$ to $V(G')$. 

Specifically, for the value $j=1$, we make \specialvertex{a}{1, \ell}{i} and  \specialvertex{\tilde{a}}{1, \ell}{i} a neighbor of \specialvertex{p}{\ell}{i} and \specialvertex{\tilde{p}}{\ell}{i}, respectively.

\item For each $j\in [\log_2(\nicefrac{n}{2})]$ and $\ell \in [\gamma_j]$, we add vertices $\specialvertex{b}{j,\ell}{i}$ and $\specialvertex{\tilde{b}}{j,\ell}{i}$ to $V(G')$. 
 
Moreover, for $j\in [\log_2(\nicefrac{n}{2})-1]$, we make \specialvertex{b}{j,\ell}{i} a neighbor of \specialvertex{a}{j, 2\ell-1}{i}, \specialvertex{a}{j, 2\ell}{i}, and \specialvertex{a}{j+1, \ell}{i}. Symmetrically, we make \specialvertex{\tilde{b}}{j,\ell}{i} a neighbor of  \specialvertex{\tilde{a}}{j, 2\ell-1}{i}, \specialvertex{\tilde{a}}{j, 2\ell}{i}, and \specialvertex{\tilde{a}}{j+1,\ell}{i}. For the special case, when $j = \log_2(\nicefrac{n}{2})$, \specialvertex{b}{j,1}{i} is a neighbor of \specialvertex{a}{j,1}{i} and \specialvertex{a}{j,2}{i}; and \specialvertex{\tilde{b}}{j,1}{i} is a neighbor of \specialvertex{\tilde{a}}{j,1}{i} and \specialvertex{\tilde{a}}{j,2}{i}. 


\medskip
\noindent For each  $\{i,j\} \subseteq [k]$, where $i<j$, we do as follows.

\item For each $\ell \in [\rho_1]$, we add vertices \specialvertex{q}{\ell}{ij} and \specialvertex{\tilde{q}}{\ell}{ij} to $V(G')$. 

Moreover, let $e$ and $e'$ denote the $2\ell-1^{st}$ and $2\ell^{th}$ elements of $E_{ij}$, respectively. Then,  \specialvertex{q}{\ell}{ij} is a neighbor of  \basevertex{e}{} and $\basevertex{e'}{}$; and symmetrically \specialvertex{\tilde{q}}{\ell}{ij} is a neighbor of \basevertex{\tilde{e}}{} and \basevertex{\tilde{e'}}{} in $G'$.

\item For each  $h\in [\log_2 (\nicefrac{m}{2})]$, and $\ell \in [\rho_h]$, we add vertices $\specialvertex{c}{h,\ell}{ij}$ and \specialvertex{\tilde{c}}{h\ell}{ij} to $V(G')$.
Moreover, for $\ell \in [\rho_1]$, $\specialvertex{c}{1,\ell}{ij}$ is a neighbor of  \specialvertex{q}{\ell}{ij}, and symmetrically \specialvertex{\tilde{c}}{1,\ell}{ij} is a neighbor of \specialvertex{\tilde{q}}{\ell}{ij} in $G'$.

\item For each $h\in [\log_2 (\nicefrac{m}{2})]$ and $\ell \in [\tau_h]$,  we add vertices \specialvertex{d}{h,\ell}{ij} and \specialvertex{\tilde{d}}{h,\ell}{ij} to $G'$.

Moreover, when $h\in [\log_2 (\nicefrac{m}{2})-1]$, \specialvertex{d}{h,\ell}{ij} is a neighbor of \specialvertex{c}{h, 2\ell-1}{ij}, \specialvertex{c}{h, 2\ell}{ij}, and \specialvertex{c}{h+1, \ell}{ij}; and symmetrically, \specialvertex{\tilde{d}}{h,\ell}{ij} is a neighbor of \specialvertex{\tilde{c}}{h, 2\ell-1}{ij}, \specialvertex{\tilde{c}}{h, 2\ell}{ij}, and \specialvertex{\tilde{c}}{h+1, \ell}{ij} in $G'$. 

For the special case, when $h=\log_2 (\nicefrac{m}{2})$, \specialvertex{d}{h,1}{ij} is a neighbor of \specialvertex{c}{h, 1}{ij} and \specialvertex{c}{h, 2}{ij}; and symmetrically, \specialvertex{\tilde{d}}{h,1}{ij} is a neighbor of \specialvertex{\tilde{c}}{h, 1}{ij}, \specialvertex{\tilde{c}}{h, 2}{ij}  in $G'$.
\end{itemize}





Figure \ref{fig:hardness_degree} illustrates the construction of $G'$. The preference list of each vertex in $G'$ is presented in Table~\ref{pref_list2}.



\begin{table*}[!t]
      \centering
   \hspace{-0.65cm} For each vertex $u \in V_{i}$, where $i\in [k]$, we have the following preferences:\\
\begin{tabularx}{\textwidth}{ C{1cm}  C{3cm}  C{13cm}}
\basevertex{u}{1}: & \pref{\basevertex{u}{2},\specialvertex{p}{\roof{\ell/2}}{i}} & where for some $\ell \in [n]$, $u$ is the $\ell^{th}$ vertex in $V_{i}$.\\ 
 \basevertex{u}{2h+1}:  & \pref{\basevertex{u}{2h},\basevertex{e}{},u_{2h+2}} & where $e$ is the $h^{th}$ element of $E_{u}$,  $h\in [r]$\\ 
 \basevertex{u}{2h}: & \pref{\basevertex{u}{2h-1}, \basevertex{u}{2h+1}} & where $h \in [r]$ \\

 \basevertex{u}{2r+2}: & \pref{\basevertex{u}{2r+1},\specialvertex{\tilde{p}}{\roof{\nicefrac{\ell}{2}}}{i}} & where for some $\ell \in [n]$, $u$ is the $\ell^{th}$ vertex in $V_{i}$ \\
 \end{tabularx}
 \medskip\\
\hspace{-0.25cm} For the special vertices of the $i^{th}$ {\it vertex gadget}, we have the following preferences:\\
\begin{tabularx}{\textwidth}{ C{1cm}  C{3cm}  C{13cm}}     
\specialvertex{p}{\ell}{i}: &\pref{\basevertex{u}{1},\basevertex{v}{1},\specialvertex{a}{1,\ell}{i}}& where for some $\ell \in [\nicefrac{n}{2}]$, $u$ and $v$ are the \\
 & &$2\ell-1^{st}$ and $2\ell^{th}$ vertices of $V_{i}$, respectively. \\
  \specialvertex{\tilde{p}}{\ell}{i}: & \pref{\basevertex{u}{2r+2},\basevertex{v}{2r+2} ,\specialvertex{\tilde{a}}{1,\ell}{i}}  & where for some $\ell \in [\nicefrac{n}{2}]$, $u$ and $v$ are the\\
  & & $2\ell-1^{st}$ and $2\ell^{th}$ vertices of $V_{i}$, respectively. \\
\specialvertex{a}{1\ell}{i}: &\pref{\specialvertex{p}{\ell}{i}, \specialvertex{b}{1,\roof{\ell/2}}{i}} & where $\ell \in [\nicefrac{n}{2}]$\\
 \specialvertex{\tilde{a}}{1,\ell}{i}:  & \pref{\specialvertex{\tilde{p}}{\ell}{i}, \specialvertex{\tilde{b}}{1,\roof{\nicefrac{\ell}{2}}}{i}} & where $ \ell \in [\nicefrac{n}{2}]$ \\
  \specialvertex{a}{j,\ell}{i}: &\pref{\specialvertex{b}{j-1, \ell}{i}, \specialvertex{b}{j, \roof{\nicefrac{\ell}{2}}}{i}}  & where $ j\in [\log_2(\nicefrac{n}{2})] \setminus \{1\}$ and $\ell \in [\nicefrac{n}{2^j}]$ \\
\specialvertex{\tilde{a}}{j,\ell}{i}: & \pref{\specialvertex{\tilde{b}}{j-1,\ell}{i}, ~\specialvertex{\tilde{b}}{j,\roof{\ell/2}}{i}} & where $j\in[\log_2(\nicefrac{n}{2})] \setminus \{1\}$ and $\ell \in [\nicefrac{n}{2^j}]$ \\
 \specialvertex{b}{j,\ell}{i}: & \pref{\specialvertex{a}{j, 2\ell -1}{i}, ~\specialvertex{a}{j, 2\ell}{i},~ \specialvertex{a}{j+1,\ell}{i}} & where $ j\in [\log_2(\nicefrac{n}{2})\!-\!1]$ and $\ell \in [\nicefrac{n}{2^{j+1}}]$\\
 \specialvertex{\tilde{b}}{j\ell}{i}: & \pref{\specialvertex{\tilde{a}}{j(2\ell -1)}{i},~ \specialvertex{\tilde{a}}{j, 2\ell}{i},~\specialvertex{\tilde{a}}{j+1, \ell}{i}} & where $j\in [\log_2(\nicefrac{n}{2})\! - \!1]$ and $\ell \in [\nicefrac{n}{2^{j+1}}]$\\
\specialvertex{b}{j,1}{i}: & \pref{\specialvertex{a}{j,1}{i},~ \specialvertex{a}{j,2}{i}} & where $ j= \log_2(\nicefrac{n}{2})$\\
 \specialvertex{\tilde{b}}{j,1}{i}: & \pref{\specialvertex{\tilde{a}}{j,1}{i}, ~\specialvertex{\tilde{a}}{j,2}{i} } & where $ j= \log_2(\nicefrac{n}{2})$
 \end{tabularx}
\medskip\\
\hspace{-1.2cm} For each edge $e\in E_{ij}$, $1\leq i<j\leq k$, we have the following preferences:\\
\begin{tabularx}{\textwidth}{ C{1cm}  C{2cm}  C{13cm}} 
\basevertex{e}{}: & \pref{\basevertex{\tilde{e}}{}, ~\basevertex{u}{2h+1},~\basevertex{v}{2h'+1} , \specialvertex{q}{\roof{\ell/2}}{ij}} & where for some $\ell \in [m]$, edge $e=uv=E_{ij}(\ell)$ and\\
 & &  for some $h,h' \in  [r]$, $e=E_{u}(h)$ and $e=E_{v}(h')$. \\
\basevertex{\tilde{e}}{}: & \pref{\basevertex{e}{}, ~\tilde{q}_{\roof{\ell/2}}^{ij}} & where for some $ \ell \in [m]$, edge $e=uv$ is the $\ell^{th}$ element of $E_{ij}$\\
     \end{tabularx}
 \medskip\\
\hspace{-0.25cm} For the special vertices of the $ij^{th}$ edge gadget, we have the following preferences:\\
\begin{tabularx}{\textwidth}{ C{1cm}  C{3cm}  C{13cm}}     
  \specialvertex{q}{\ell}{ij}: & \pref{ \basevertex{e}{}, \basevertex{e'}{} ,\specialvertex{c}{1,\ell}{ij} }  & where for some $ \ell \in [\nicefrac{m}{2}]$, edges $e$ and $e'$ are the\\
  & & $2\ell-1^{st}$ and $2\ell^{th}$ elements of $E_{ij}$, respectively. \\
 \specialvertex{\tilde{q}}{\ell}{ij}: & \pref{\basevertex{\tilde{e}}{}, \basevertex{\tilde{e'}}{}, \specialvertex{\tilde{c}}{1,\ell}{ij}}  & where for some $ \ell \in [\nicefrac{m}{2}]$, edges $e$ and $e'$ are the\\
  & & $2\ell-1^{st}$ and $2\ell^{th}$ elements of $E_{ij}$, respectively. \\
\specialvertex{c}{1,\ell}{ij}: &  \pref{ \specialvertex{q}{\ell}{ij}, \specialvertex{d}{1, \roof{\ell/2}}{ij}} & where $ \ell \in [\nicefrac{m}{2}]$ \\
  \specialvertex{\tilde{c}}{1,\ell}{ij}: & \pref{\specialvertex{\tilde{q}}{\ell}{ij},~ \specialvertex{\tilde{d}}{1, \roof{\ell/2}}{ij}} & where $\ell \in [\nicefrac{m}{2}]$ \\
 \specialvertex{c}{h,\ell}{ij}: & \pref{\specialvertex{d}{h-1, \ell}{ij}, ~\specialvertex{d}{h, \roof{\ell/2}}{ij}} & where $ h \in[\log_2(\nicefrac{m}{2})] \setminus \{1\}, \ell \in [\nicefrac{m}{2^{h}}]$ \\

 \specialvertex{\tilde{c}}{h,\ell}{ij}: & \pref{\specialvertex{\tilde{d}}{h-1, \ell}{ij}, \specialvertex{\tilde{d}}{h, \roof{\ell/2}}{ij}} & where $ h \in[\log_2(\nicefrac{m}{2})] \setminus \{1\}$ and $\ell \in [\nicefrac{m}{2^{h}}]$ \\
\specialvertex{d}{h,\ell}{ij}:  &  \pref{\specialvertex{c}{h, 2\ell -1}{ij}, ~ \specialvertex{c}{h, 2\ell}{ij},~ \specialvertex{c}{h+1, \ell}{ij}} & where $ h \in [ \log_2(\nicefrac{m}{2}) \! - \! 1]$ and $\ell \in [\nicefrac{m}{2^{h+1}}]$\\
\specialvertex{\tilde{d}}{h, \ell}{ij}: & \pref{\specialvertex{\tilde{c}}{h, 2\ell -1}{ij},~ \specialvertex{\tilde{c}}{h, 2\ell}{ij},~ \specialvertex{\tilde{c}}{h+1, \ell}{ij}} & where $ h \in [ \log_2(\nicefrac{m}{2}) \!- \!1]$ and $\ell \in [\nicefrac{m}{2^{h+1}}]$\\
\specialvertex{d}{h,1}{ij}: &  \pref{\specialvertex{c}{h,1}{ij},~ \specialvertex{c}{h,2}{ij} } & where $ h = \log_2(\nicefrac{m}{2})$\\
 \specialvertex{\tilde{d}}{h, 1}{ij}: & \pref{\specialvertex{\tilde{c}}{h, 1}{ij},~ \specialvertex{\tilde{c}}{h,2}{ij}} & where $h = \log_2(\nicefrac{m}{2})$\\
    \end{tabularx}
   \caption{Preference lists in the proof of Theorem~\ref{thm:asm}; notation $\pref{\cdot, \cdot}$ denotes the order of preference over neighbors.}
   \label{pref_list2}
\end{table*}

\par
{\bf Parameter$\colon$} We set $k'=k+\nicefrac{k(k-1)}{2}$, and $t=k'$. 
\par
Clearly, this construction can be carried out in polynomial time. 
Next, we will prove that the graph $G'$ is bipartite.

 


 

\begin{supress}

 \il{From here gone to main body}
In this section we prove Theorem~\ref{thm:asm}. We give a polynomial-time parameter preserving many-to-one reduction from the \WOH problem \mcq (\mcqsmall) (\cite{cai2008parameterized}) on regular graphs. \hide{in which we are given a regular graph $G=(V,E)$ and a partition of $V(G)$ into $k$ parts, $V_1,\ldots,V_k$, and the objective is to decide if there exists a subset $S\subseteq V(G)$ such that $| S\cap V_i| =1$, for each $i\in [k]$, and the induced subgraph $G[S]$ is a clique. }

It will be necessary for us to assume that certain sets are ordered. This ordering uniquely defines the $t^{th}$ element of the set (for an appropriately defined value of $t$), and thereby enables us to refer to the $t^{th}$ element of the set unambiguously. We assume that sets $V_{i}$ (for each $i\in [k]$) and $E_{ij}$ (for each $\{i,j\} \sse [k], i<j$) have a canonical order, and thus for an appropriately defined value $t$, $V_{i}(t)$ ($E_{ij}(t)$) and $\sigma(V_{i}, v)$ ($\sigma(E_{ij}, e)$) are uniquely defined. For ease of exposition, for any vertex $v\in V(G')$ we will refer to its set of neighbors, as an ordered set. In such a situation we will denote $N(v)=\pref{ \cdot, \cdot}$. 

\hide{For any ordered set $X$, and an appropriately defined value $t$, $X(t)$ denotes the $t^{th}$ element of the set $X$. Conversely, suppose that $x$ is $h^{th}$ element of the set $X$, then $\sigma(x, X)=h$.} 

Given an instance $\Co{I}=(G,(V_1,\ldots,V_k))$ of \mcqsmall, where $G$ is a regular graph whose degree is denoted by $r$, we will next describe the construction of an instance $\Co{J}=(G',\Co{L},k',t)$ of \asm. 

  \par
 {\bf Construction.} We begin by introducing some notations. For any $\{i, j\} \sse [k]$, such that $i<j$, we use $E_{ij}$ to denote the set of edges between sets $V_i$ and $V_j$. For each $i\in [k]$, we may assume that $| V_i| =n=2^p$, and for each $\{i,j\} \subseteq [k]$, we may assume that $\lvert E_{ij} \rvert =m=2^{p'}$, for some positive integers $p$ and $p'$ greater than one.\footnote{Let $p$ be the smallest positive integer greater than one such that $n<2^p$, add $2^p-n$ isolated vertices in $V_i$. Similarly, let $p'$ be the smallest positive integer greater than one such that $m<2^{p'}$, add  $2^{p'}-m$ isolated edges (an edge whose endpoints are of degree exactly one) to $E_{ij}$. Note that if $(G,(V_1,\ldots,V_k))$ was a \WOH instance of \mcqsmall earlier, then so even now.}.


 
 
For each  $j\in [\log_2 (\nicefrac{n}{2})]$, let $\beta_j= \nicefrac{n}{2^j}$, and $\gamma_j=\nicefrac{n}{2^{j+1}}$. For each  $j\in [\log_2(\nicefrac{m}{2})] $, let $\rho_j=\nicefrac{m}{2^j}$, and $\tau_j= \nicefrac{m}{2^{j+1}}$. 
Next, we are ready to describe the construction of the graph $G'$.

\begin{itemize}
\item[ ]~~~
\noindent\hspace{-0.7cm}{\bf Base vertices:}
\item For each vertex $u \in V(G)$, we have $2r+2$ vertices in $G'$, denoted by $\{\basevertex{u}{i} : i\in [2r+2]\}$, connected via a path: $(\basevertex{u}{1},\ldots,\basevertex{u}{2r+2})$. 

 \item For each edge $e\in E(G)$, we have vertices  $\basevertex{e}{}$ and $\basevertex{\tilde{e}}{}$ in $G'$ that are neighbors. 
 
 
\hide{For any $u\in V(G)$, we define $\X{E}_{u}=\{e \in V(G') | e \in E_{u}\}$ and assume that the set $\X{E}_{u}$ has the same canonical order as $E_{u}$.} 
 

\item  For each $h\in [r]$, \basevertex{u}{2h+1} is a neighbor of the vertex \basevertex{e}{}, where $e=\sigma(E_{u}, h)$. 

\smallskip
\noindent\hspace{-0.3cm}{\bf Special vertices.} For each $i\in [k]$, we define a set of special vertices as follows. 
\item For each $\ell \in [\beta_1]$, we add vertices \specialvertex{p}{\ell}{i} and \specialvertex{\tilde{p}}{\ell}{i} to $V(G')$. Let $u$ and $v$ denote the $2\ell-1^{st}$ and the $2\ell^{th}$ vertices in $V_{i}$, respectively. Then, the vertex 
\specialvertex{p}{\ell}{i} is a neighbor of vertices \basevertex{u}{1} and \basevertex{v}{1}; and the vertex \specialvertex{\tilde{p}}{\ell}{i} is a neighbor of vertices \basevertex{u}{2r+2} and \basevertex{v}{2r+2} in $G'$.
 
\item For each $j\in [\log_2 (\nicefrac{n}{2})]$ and $\ell \in [\beta_j]$, we add vertices $\specialvertex{a}{j,\ell}{i}$ and $\specialvertex{\tilde{a}}{j,\ell}{i}$ to $V(G')$. 

Specifically, for the value $j=1$, we make \specialvertex{a}{1, \ell}{i} and  \specialvertex{\tilde{a}}{1, \ell}{i} a neighbor of \specialvertex{p}{\ell}{i} and \specialvertex{\tilde{p}}{\ell}{i}, respectively.

\item For each $j\in [\log_2(\nicefrac{n}{2})]$ and $\ell \in [\gamma_j]$, we add vertices $\specialvertex{b}{j,\ell}{i}$ and $\specialvertex{\tilde{b}}{j,\ell}{i}$ to $V(G')$. 
 
Moreover, for $j\in [\log_2(\nicefrac{n}{2})-1]$, we make \specialvertex{b}{j,\ell}{i} a neighbor of \specialvertex{a}{j, 2\ell-1}{i}, \specialvertex{a}{j, 2\ell}{i}, and \specialvertex{a}{j+1, \ell}{i}. Symmetrically, we make \specialvertex{\tilde{b}}{j,\ell}{i} a neighbor of  \specialvertex{\tilde{a}}{j, 2\ell-1}{i}, \specialvertex{\tilde{a}}{j, 2\ell}{i}, and \specialvertex{\tilde{a}}{j+1,\ell}{i}. For the special case, when $j = \log_2(\nicefrac{n}{2})$, \specialvertex{b}{j,1}{i} is a neighbor of \specialvertex{a}{j,1}{i} and \specialvertex{a}{j,2}{i}; and \specialvertex{\tilde{b}}{j,1}{i} is a neighbor of \specialvertex{\tilde{a}}{j,1}{i} and \specialvertex{\tilde{a}}{j,2}{i}. 


\medskip
\noindent For each  $\{i,j\} \subseteq [k]$, where $i<j$, we do as follows.

\item For each $\ell \in [\rho_1]$, we add vertices \specialvertex{q}{\ell}{ij} and \specialvertex{\tilde{q}}{\ell}{ij} to $V(G')$. 

Moreover, let $e$ and $e'$ denote the $2\ell-1^{st}$ and $2\ell^{th}$ elements of $E_{ij}$, respectively. Then,  \specialvertex{q}{\ell}{ij} is a neighbor of  \basevertex{e}{} and $\basevertex{e'}{}$; and symmetrically \specialvertex{\tilde{q}}{\ell}{ij} is a neighbor of \basevertex{\tilde{e}}{} and \basevertex{\tilde{e'}}{} in $G'$.

\item For each  $h\in [\log_2 (\nicefrac{m}{2})]$, and $\ell \in [\rho_h]$, we add vertices $\specialvertex{c}{h,\ell}{ij}$ and \specialvertex{\tilde{c}}{h\ell}{ij} to $V(G')$.
Moreover, for $\ell \in [\rho_1]$, $\specialvertex{c}{1,\ell}{ij}$ is a neighbor of  \specialvertex{q}{\ell}{ij}, and symmetrically \specialvertex{\tilde{c}}{1,\ell}{ij} is a neighbor of \specialvertex{\tilde{q}}{\ell}{ij} in $G'$.

\item For each $h\in [\log_2 (\nicefrac{m}{2})]$ and $\ell \in [\tau_h]$,  we add vertices \specialvertex{d}{h,\ell}{ij} and \specialvertex{\tilde{d}}{h,\ell}{ij} to $G'$.

Moreover, when $h\in [\log_2 (\nicefrac{m}{2})-1]$, \specialvertex{d}{h,\ell}{ij} is a neighbor of \specialvertex{c}{h, 2\ell-1}{ij}, \specialvertex{c}{h, 2\ell}{ij}, and \specialvertex{c}{h+1, \ell}{ij}; and symmetrically, \specialvertex{\tilde{d}}{h,\ell}{ij} is a neighbor of \specialvertex{\tilde{c}}{h, 2\ell-1}{ij}, \specialvertex{\tilde{c}}{h, 2\ell}{ij}, and \specialvertex{\tilde{c}}{h+1, \ell}{ij} in $G'$. 

For the special case, when $h=\log_2 (\nicefrac{m}{2})$, \specialvertex{d}{h,1}{ij} is a neighbor of \specialvertex{c}{h, 1}{ij} and \specialvertex{c}{h, 2}{ij}; and symmetrically, \specialvertex{\tilde{d}}{h,1}{ij} is a neighbor of \specialvertex{\tilde{c}}{h, 1}{ij}, \specialvertex{\tilde{c}}{h, 2}{ij}  in $G'$.
\end{itemize}





Figure \ref{fig:hardness_degree} illustrates the construction of $G'$. The preference list of each vertex in $G'$ is presented in Table~\ref{pref_list2}.


\par
{\bf Parameter$\colon$} We set $k'=k+\nicefrac{k(k-1)}{2}$, and $t=k'$. 
\par
Clearly, this construction can be carried out in polynomial time. Next, we will prove 
that the graph $G'$ is bipartite. 

\il{Upto here in teh main body}

*********************************************************
\end{supress}


\begin{figure*}
    \centering
\hspace{-1.1cm}\includegraphics[width=13cm,height=15cm,keepaspectratio,]{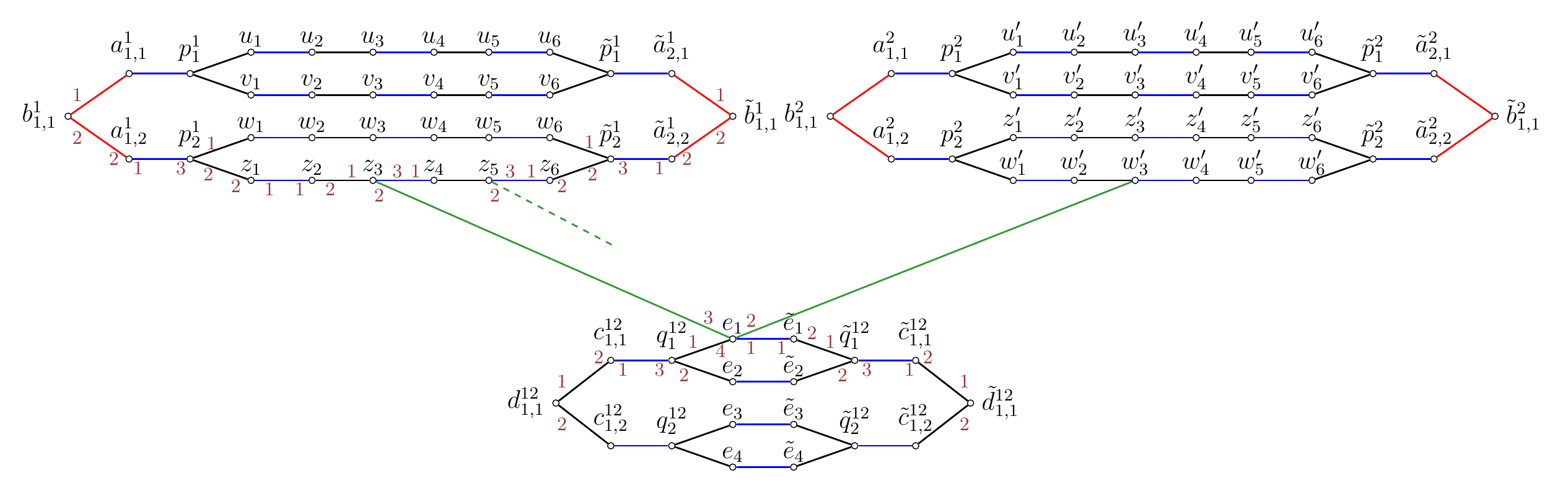}
\caption{An illustration of construction of graph $G'$ in \W[1]-hardness of \asm for constant sized preference list. Here, {\color{blue} blue} colored edges belongs to the stable matching $\mu$. Here, $n=4$, $m=4$, and $r=2$}\label{fig:hardness_degree}
\end{figure*}




\begin{clm}\label{cl:G'-bipartite}
Graph $G'$ is bipartite.
\end{clm}

\begin{proof} 

We show that $G'$ is a bipartite graph by creating a bipartition $(X,Y)$ for $G'$ as follows. We define the following sets.
\[A=\{a_{j,\ell}^i : i\in[k], j\in [\log_2(\nicefrac{n}{2})], \ell \in [\beta_j]\}\]
 \[\tilde{A}=\{\tilde{a}_{j,\ell}^i :
  i\in[k], j\in [\log_2(\nicefrac{n}{2})], \ell \in [\beta_j]\}\] 
\[B=\{b_{j,\ell}^i : i\in[k], j\in [\log_2(\nicefrac{n}{2})], \ell \in [\gamma_j]\}\] 
 \[\tilde{B}=\{\tilde{b}_{j,\ell}^i : i\in[k], j\in [\log_2(\nicefrac{n}{2})], \ell \in [\gamma_j]\}\] 
 We add $A$ to $X$ and $\tilde{A}$ to $Y$. Note that there is no edge between the vertices in $A$ (or $\tilde{A}$). Since no vertex of $B$(or $\tilde{B}$) is adjacent to $\tilde{A}$(or $A$), we add $B$ to $Y$ and $\tilde{B}$ to $X$. 
 
 Let $P=\{p_{\ell}^i : i\in [k], \ell \in [\beta_1]\}$ and  $\tilde{P}=\{\tilde{p}_{\ell}^i : i\in [k], \ell \in [\beta_1]\}$. We add $P$ to $Y$ and $\tilde{P}$ to $X$. 
We define the following sets of vertices. 
  \[U_{odd}=\{\basevertex{u}{2h-1} : u\in V_{i}, i\in [k], h\in [r+1]\} \text{ and }\]  \[ U_{even}=\{ \basevertex{u}{2h} : u \in V_{i}, i\in [k],  h\in [r+1]\}\]
    We add $U_{odd}$ to $X$ and $U_{even}$ to $Y$. 
    We define the following sets.
    \[E_1=\{ \basevertex{e}{} : e \in E_{ij}, \{i,j\} \sse [k] \} \text{ and }\] \[ E_2=\{\basevertex{\tilde{e}}{} : e \in E_{ij}, \{i,j\} \sse [k] \}\]
      We add $E_1$ to $Y$ and $E_2$ to $X$. 

   We next  define the following sets. 
      \[Q=\{q_{\ell}^{ij} :  \{i,j\} \sse [k], i<j, \ell \in [\rho_1]\}  \text{ and }\] \[ \tilde{Q}=\{\tilde{q}_{\ell}^{ij} : \{i,j\} \sse [k], i<j, \ell \in [\rho_1]\}\]
        We add $Q$ to $X$ and $\tilde{Q}$ to $Y$. 
      Again define the the following two sets.
       \[C=\{c_{h, \ell}^{ij} : \{i,j\} \sse [k], i<j, h \in [\log_2 (\nicefrac{m}{2})], \ell \in [\rho_h]\}\]  \[ \tilde{C}=\{\tilde{c}_{h, \ell}^{ij} : \{i,j\} \sse [k], i<j, h \in [\log_2 (\nicefrac{m}{2})], \ell \in [\rho_h]\}\]
        We add $C$ to $Y$ and $\tilde{C}$ to $X$.  
     Finally we define the sets,
      \[D=\{d_{h,\ell}^{ij} :  \{i,j\} \sse [k], i<j, h \in [\log_2 (\nicefrac{m}{2})], \ell \in [\tau_h]\} \]
      \[ \tilde{D}=\{\tilde{d}_{h, \ell}^{ij} :  \{i,j\} \sse [k], i<j, h \in [\log_2 (\nicefrac{m}{2})], \ell \in [\tau_h]\}\]
       We add $D$ to $X$ and $\tilde{D}$ to $Y$.       
      Observe that $X$ and $Y$ are independent sets in $G'$. Hence, $G'$ is a bipartite graph. Figure~\ref{fig:bipartite} illustrates this bipartition of the graph $G'$. 
\end{proof}

\begin{figure*}
    \centering
\includegraphics[width=6cm,height=4cm,keepaspectratio]{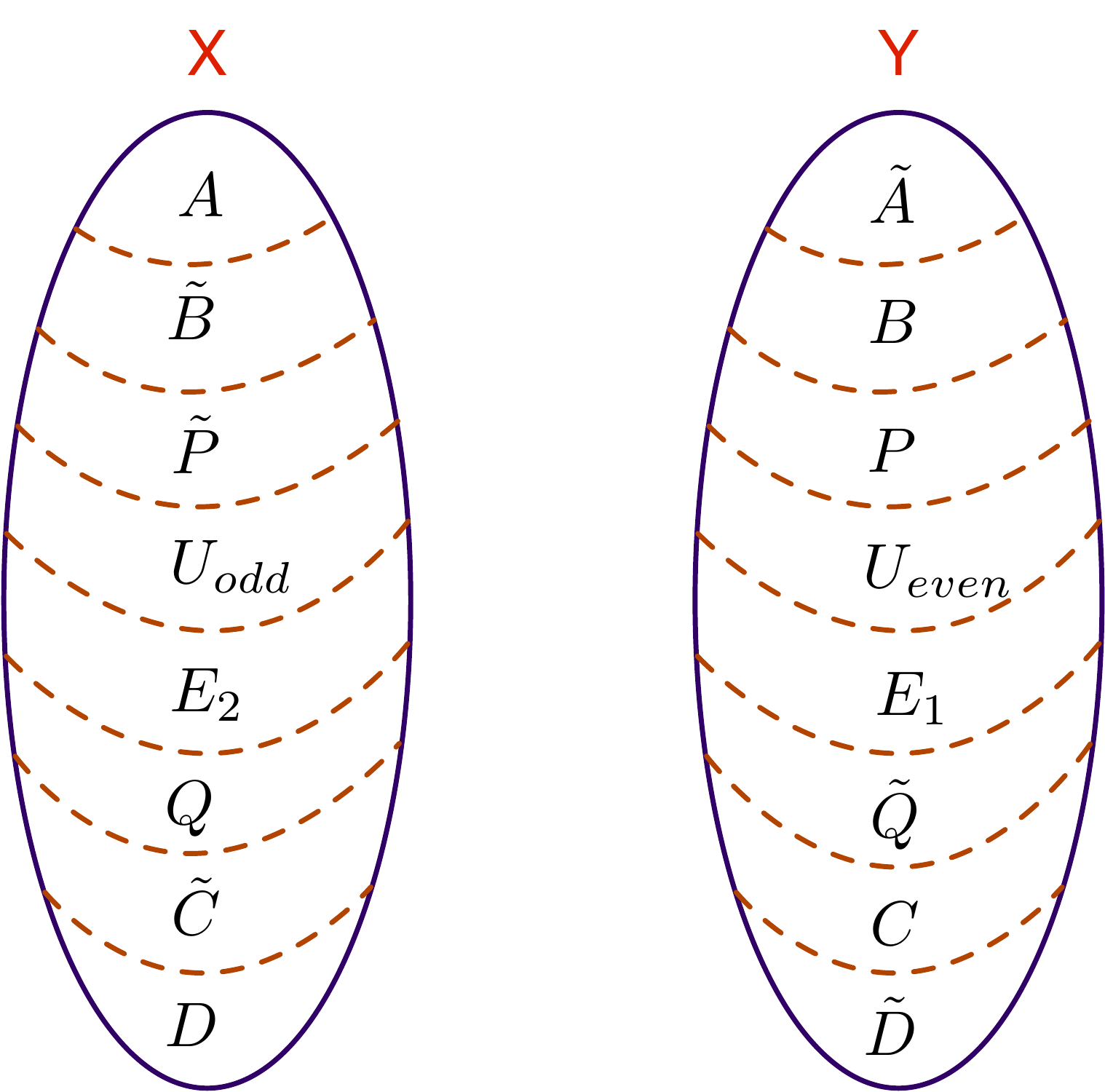}
\caption{A bipartition of the graph $G'$, constructed in the  \textsf{\textup{W[1]-hardness}} of \asm.}\label{fig:bipartite}
\end{figure*}

This completes the construction of an instance of \asm. 

\medskip

\noindent {\bf Correctness:} Since we are interested in a matching which is at least $t$ more than the size of a stable matching, we need to know the size of a stable matching. Towards this, we construct a stable matching $\mu$ that contains the following set of edge 
 
\[ \left(\cup_{\substack{u \in V(G)}}\{ \basevertex{u}{2h-1}\basevertex{u}{2h} \in E(G'): h\in [r+1] \}\right) \bigcup \left(\cup_{\substack{e\in E(G)}}\{ \basevertex{e}{} \basevertex{\tilde{e}}{} \in E(G')\}\right) \tag{I} \label{size-mu}\]

Additionally, for each $i \in [k]$ and $\ell \in [\nicefrac{n}{2}]$, we add $a_{1,\ell}^ip_{\ell}^i$ and $\tilde{a}_{1,\ell}^i\tilde{p}_{\ell}^i$ to $\mu$.\hide{thus, $kn$ edges are added.} For each $i\in [k]$, $j\in [\log_2(\nicefrac{n}{2})] \setminus \{1\}$, and $\ell \in [\beta_j]$, we add $a_{j, \ell}^ib_{j-1, \ell}^i$ and  $\tilde{a}_{j, \ell}^i\tilde{b}_{j-1,\ell}^i$ to $\mu$. For each $\{i,j\} \subseteq [k]$, $i<j$, and $\ell \in [\nicefrac{m}{2}]$, we add $c_{1,\ell}^{ij}q_{\ell}^{ij}$ and $\tilde{c}_{1,\ell}^{ij}\tilde{q}_{\ell}^{ij}$ to $\mu$. \hide{thus, $m(\nicefrac{k(k-1)}{2})$ edges are added.} For each $\{i,j\} \subseteq [k]$, $i<j$, $h\in [\log_2(\nicefrac{m}{2})] \setminus \{1\}$, and $\ell \in [\rho_h]$, we add $c_{h,\ell}^{ij}d_{h-1, \ell}^{ij}$ and $\tilde{c}_{h,\ell}^{ij}\tilde{d}_{(h-1)\ell}^{ij}$ to $\mu$.\hide{so $(m-4)(\nicefrac{k(k-1)}{2})$ edges are added.} 
This completes the construction of the matching $\mu$. 

\begin{clm}\label{size-of-mu}
Matching $\mu$ has size $kn(r+1)+\nicefrac{mk(k-1)}{2}+2k(n-2)+k(k-1)(m-2)$. Furthermore, $\mu$ is a stable matching in $G'$. 
\end{clm}

\begin{proof}Due to \Cref{size-mu}, we know that $\mu$ contains at least $kn(r+1)+\nicefrac{mk(k-1)}{2}$ edges because $|V_{i}|=n$ for each $i\in [k]$ and $|E_{ij}| =m$ for each $\{i,j\} \sse [k]$. The other edges added to $\mu$ can be counted separately, leading to the following relation. 

\begin{align*}
| \mu| = & ~kn(r+1)+\nicefrac{mk(k-1)}{2}+ kn+ k(n-4)+ (\nicefrac{mk(k-1)}{2}) + (m-4)(\nicefrac{k(k-1)}{2})\\
=& ~kn(r+1)+\nicefrac{mk(k-1)}{2}+2k(n-2)+k(k-1)(m-2)
\end{align*}

%
%
%

Next, to show that $\mu$ is a stable matching in $G'$, we will exhaustively argue for each vertex in $G'$ that there is no blocking edge incident to it. 

We begin by noting that for any vertex $u \in V(G)$, vertices $\basevertex{u}{1}$ and $\basevertex{u}{2}$ in $G'$ prefer each other over any other vertex in $G'$. Therefore, edge $\basevertex{u}{1} \basevertex{u}{2}$ is a static edge and must belong to every stable  matching in $G'$. Similarly, for any $e\in E(G)$, we note that $\basevertex{e}{}\basevertex{\tilde{e}}{}$ is a static edge in $G'$, and thus belongs to every stable matching in $G'$. For any $u\in V(G)$ and $h\in [r]$, we know that vertex $\basevertex{u}{2h+1}$ is the first preference of \basevertex{u}{2h+2}. Thus, there cannot exist a blocking edge incident to \basevertex{u}{2h+2}, where $h\in [r]$. 
Moreover, for any $h\in [r]$, the vertices that $\basevertex{u}{2h+1}$ prefers over $\basevertex{u}{2h+2}$ are matched to their top preferences. Consequently,  there cannot be a blocking edge incident to $\basevertex{u}{2h+1}$. 

Since for each $u\in V_{i}$, $i\in [k]$, vertices $u_{1}$ and  $u_{2r+2}$ are matched to their top preferences respectively, thus for any $\ell\in [\nicefrac{n}{2}]$ the edges $\basevertex{u}{1} \specialvertex{p}{\ell}{i}$ and $\basevertex{u}{2r+2} \specialvertex{\tilde{p}}{\ell}{i}$ cannot be a blocking edge \wrt $\mu$. Thus, there is no blocking edge incident to $\specialvertex{p}{\ell}{i}$ and $\specialvertex{\tilde{p}}{\ell }{i}$, for  $\ell\in [\nicefrac{n}{2}]$. Analogously, we can argue that there is no blocking edge incident on $\specialvertex{q}{h}{ij}$ and $\specialvertex{\tilde{q}}{h}{ij}$, for any $\{i,j\} \sse [k]$, $i<j$ and $h\in [\nicefrac{m}{2}]$.

Since for each $i\in [k]$, $j\in [\log_2 (\nicefrac{n}{2})]$, and $\ell \in [\beta_j]$, vertices $\specialvertex{a}{j,\ell}{i}$ and $\specialvertex{\tilde{a}}{j,\ell}{i}$ are matched to their top preferences respectively, there is no blocking edge incident to $\specialvertex{a}{j,\ell}{i}$ or $\specialvertex{\tilde{a}}{j,\ell}{i}$. Analogously, there is no blocking edge incident on $\specialvertex{c}{h,\ell}{ij}$ or $\specialvertex{\tilde{c}}{h,\ell}{ij}$, for any $\{i,j\}\subseteq [k], i<j, h\in [ \log_2(\nicefrac{m}{2})]$, and $\ell \in [\rho_h]$. 

For any $i\in [k]$, $j\in [\log_2 (\nicefrac{n}{2})]$, and $\ell \in [\gamma_j]$, vertices that $\specialvertex{b}{j, \ell}{i}$ prefers over $\specialvertex{a}{j+1,\ell}{i}$ (i.e., $\specialvertex{a}{j, 2\ell-1}{i}$ and $\specialvertex{a}{j, 2\ell}{i}$) are matched to their top preferences respectively, there is no blocking edge incident to $\specialvertex{b}{j,\ell}{i}$. By symmetry, there is no blocking edge incident to $\specialvertex{\tilde{b}}{h,\ell}{i}$. Analogously, there is also no blocking edge incident to $\specialvertex{d}{h',\ell'}{ij}$, or $\specialvertex{\tilde{d}}{h',\ell'}{ij}$, for any $\{i,j\}\subseteq [k], i<j, h'\in [ \log_2(\nicefrac{m}{2})]$, and $\ell' \in [\tau_h]$. 

Hence, we can conclude that  $\mu$ is a stable matching in $G'$.
\end{proof}

Next, we will formally prove the equivalence between the instance of \mcqsmall and the instance of \asm. In particular, we will prove the following lemma.

\begin{lemma}\label{lem:asm-hard-correctness} $\Co{I}= (G, (V_{1}, \ldots, V_{k}))$ is a \yes-instance of \mcqsmall if and only if $\Co{J}=(G', \Co{L}, \mu, k',t)$ is a \yes-instance of \asm.

\end{lemma}

Before giving the proof of Lemma~\ref{lem:asm-hard-correctness}, we give a structural property of any matching in $G'$ which will be used later.

\begin{clm}\label{tilde eta is perfect}
Let $\tilde{\eta}$ be a matching in $G'$ of size $|\mu|+t$. Then, $\tilde{\eta}$ is a perfect matching in $G'$.
\end{clm}
\begin{proof}
  We first count the number of vertices in $G'$. Note that for each vertex in $G$, we have a path  of length $2r+2$ in $G'$. Since $|V(G)|=nk$, there are $(2r+2)nk$ such vertices in $V(G')$. \hide{Next, for each $i\in [k]$, we add $\nicefrac{n}{2}$ vertices $\{ p_\ell^i,\tilde{p}_\ell^i : \ell \in [\nicefrac{n}{2}]\}$.}
  For each $i \in [k]$ and $\ell \in [\nicefrac{n}{2}]$, we added $p_\ell^i,\tilde{p}_\ell^i$.
 Hence, we have added $kn$ special vertices to $V(G')$.  Note that there are
 $2k(n-2)$ vertices in the set $\{a_{j\ell}^i, \tilde{a}_{j \ell}^i : i\in [k], j\in [\log_2(\nicefrac{n}{2})], \ell \in [\nicefrac{n}{2^j}]\}$. Additionally, we have $k(n-2)$ vertices in the set $\{b_{j,\ell}^i, \tilde{b}_{j,\ell}^i : i\in [k], j\in [\log_2(\nicefrac{n}{2})], \ell \in [\nicefrac{n}{2^{j+1}}]\}$. 
 
 Now, we count the vertices in $G'$ that corresponding to edges in $G$.  Note that for each edge in $G$, we have two vertices in $G'$. Since $|E_{ij}|=m$, where $\{i,j\}\subseteq [k]$, there are $2m(\nicefrac{k(k-1)}{2})$  vertices in the set $\{\basevertex{e}{}, \basevertex{\tilde{e}}{}: e\in E(G)\}$. There are $m(\nicefrac{k(k-1)}{2})$ vertices in the set $\{q_\ell^{ij},\tilde{q}_\ell^{ij} : \{i,j\}\subseteq [k], i<j,\ell \in [\nicefrac{m}{2}] \}$. There are $k(k-1)(m-2)$ vertices in the set $\{c_{h, \ell}^{ij}, \tilde{c}_{h,\ell}^{ij}: \{i,j\}\subseteq [k], i<j, h\in [\log_2(\nicefrac{m}{2})], \ell \in [\nicefrac{m}{2^h}]\}$. Similarly, we have $\nicefrac{k(k-1)(m-2)}{2}$ vertices in the set $\{d_{h, \ell}^{ij}, \tilde{d}_{h,\ell}^{ij} : \{i,j\}\subseteq  [k], h\in [\log_2(\nicefrac{m}{2})], \ell \in [\nicefrac{m}{2^{h+1}}]\}$. Hence, 
\[ |V(G')|=2(r+1)kn+2k(2n-3)+mk(k-1)  +2k(k-1)(m-\nicefrac{3}{2}) \]

  Recall that $\lvert \mu \rvert = (r+1)kn+\nicefrac{mk(k-1)}{2}+2k(n-2)+k(k-1)(m-2)$ and $t=k+\nicefrac{k(k-1)}{2}$. Therefore, $|\tilde{\eta}|= (r+1)kn+ \nicefrac{mk(k-1)}{2} +k(2n-3)+ k(k-1)(m-\nicefrac{3}{2})$. Hence,  $\tilde{\eta}$ is a perfect matching in $G'$. 
  \end{proof}

Now, we are ready to prove Lemma~\ref{lem:asm-hard-correctness}.
\begin{proof}[Proof of Lemma~\ref{lem:asm-hard-correctness}]

In the forward direction, let $S$ be a solution of \mcqsmall for $\Co{I}$, i.e  $|X\cap V_i|=1$, for each $i\in[k]$ and $G[X]$ is a clique in $G$. 

\noindent{\bf Defining a solution matching:} We construct a solution $\eta$ to $\Co{J}$ as follows. Initially, we set $\eta = \mu$. 
Suppose that $u=S\cap V_{i}$, then from $\eta$ we delete edges $\{\basevertex{u}{2h-1} \basevertex{u}{2h}: h\in [r+1]\}$; and add edges $\{\basevertex{u}{2h} \basevertex{u}{2h+1} : h\in [r]\}$. 

Let $\ell=\sigma(V_{i},u)$, i.e the solution $S$ contains the $\ell^{th}$ vertex of the set $V_{i}$. Then, we  delete $\{\specialvertex{a}{1, \lceil \nicefrac{\ell}{2} \rceil}{i} \specialvertex{p}{\lceil \nicefrac{\ell}{2} \rceil}{i}, \specialvertex{\tilde{a}}{1, \lceil \nicefrac{\ell}{2} \rceil}{i} \specialvertex{\tilde{p}}{\lceil \nicefrac{\ell}{2} \rceil}{i}\}$ from $\eta$ and add $\{\basevertex{u}{1} \specialvertex{p}{\lceil \nicefrac{\ell}{2} \rceil}{i}, \basevertex{u}{2r+2} \specialvertex{\tilde{p}}{\lceil \nicefrac{\ell}{2} \rceil}{i}\}$ to $\eta$. Additionally, we delete $\{\specialvertex{a}{j,h}{i} \specialvertex{b}{j-1, h}{i}, \specialvertex{\tilde{a}}{j,h}{i} \specialvertex{\tilde{b}}{j-1,h}{i} : j\in [\log_2 (\nicefrac{n}{2})] \setminus \{1\}, h=\roof{\nicefrac{\ell}{2^j}}\}$ and add 
set $\{ \specialvertex{a}{j,h}{i} \specialvertex{b}{j,h'}{i}, \specialvertex{\tilde{a}}{j,h}{i}\specialvertex{\tilde{b}}{j,h'}{i}: j\in [\log_2 (\nicefrac{n}{2})], h= \roof{ \nicefrac{\ell}{2^j}}, h'=\roof{\nicefrac{\ell}{2^{j+1}}} \}$ to $\eta$. 

Let edge $e= E(G[S])\cap E_{ij}$, \ie edge $e$ in $E_{ij}$ is in the clique solution, for some 
$\{i,j\}\subseteq [k]$. Suppose that for some $\ell \in [m]$, $e$ is the $\ell^{th}$ edge in $E_{ij}$. Then, we delete set $\{ \basevertex{e}{} \basevertex{\tilde{e}}{}, q_{\lceil \nicefrac{\ell}{2} \rceil}^{ij}c_{1, \lceil \nicefrac{\ell}{2} \rceil}^{ij}, \tilde{q}_{\lceil \nicefrac{\ell}{2} \rceil}^{ij}\tilde{c}_{1, \lceil \nicefrac{\ell}{2} \rceil}^{ij}\}$ from $\eta$ and add $\{ \basevertex{e}{} q_{\lceil \nicefrac{\ell}{2} \rceil}^{ij}, \basevertex{\tilde{e}}{} \tilde{q}_{\lceil \nicefrac{\ell}{2} \rceil}^{ij}\}$ to $\eta$. Additionally, we delete edges $\{c_{h, s}^{ij}d_{h-1, s}^{ij}, ~\tilde{c}_{h, s}^{ij}\tilde{d}_{h-1, s}^{ij}: h\in [\log_2(\nicefrac{m}{2})] \setminus \{1\} \}, s=\lceil \nicefrac{\ell}{2^h}\rceil  \}$ from $\eta$ and add set $\{c_{h,s}^{ij}d_{h, s'}^{ij}, ~\tilde{c}_{h, s}^{ij}\tilde{d}_{h, s'}^{ij}: h\in [\log_2(\nicefrac{m}{2})], s=\lceil \nicefrac{\ell}{2^h}\rceil, s'=\lceil \nicefrac{\ell}{2^{h+1}}\rceil \}$. Due to the construction of $\eta$, clearly it is a matching.

The following result implies that the matching $\eta$ constructed as above satisfies the size bound of a solution for our instance $\Co{J}$ of \asm. 


\begin{clm}\label{size-eta}
Matching $\eta$ described above has size $|\mu|+k+\nicefrac{k(k-1)}{2}$. 
\end{clm}

\begin{proof}For each (clique) vertex $u= S\cap V_{i}$, where $i\in [k]$,  we delete $r+2\log_2(\nicefrac{n}{2})+1$ edges from $\eta$ (which also belong to $\mu$), and add $r+2\log_2(\nicefrac{n}{2})+2$ edges to $\eta$. This gives us an an additional $k$ edges in $\eta$. 

Similarly, for each clique edge $e= E(G[S]) \cap E_{ij}$, where $\{i,j\} \subseteq [k], i<j$, we delete $2 \log_2(\nicefrac{m}{2}) +1$ edges from $\eta$ (which also belong to $\mu$), and add  $2 \log_2(\nicefrac{m}{2}) +2$ edges to $\eta$. This, gives us an additional $\nicefrac{k(k-1)}{2}$ edges in $\eta$. Thus, in total $|\eta|= |\mu|+k+\nicefrac{k(k-1)}{2}$. 
\end{proof}

Next, we prove that $\eta$ has $k'=k+\nicefrac{k(k-1)}{2}$ blocking edges. Due to Proposition~\ref{cl:bp vertices}, for a blocking edge with respect to $\eta$, at least one of its endpoint is in $V(\mu \triangle \eta)$. Therefore,  we only need to investigate the vertices of $V(\mu \triangle \eta)$. We begin by characterizing the vertices in the set $V(\mu \triangle \eta)$.



Note that 
\begin{multline*} 
V(\mu \triangle \eta)= 
~~\{u_{2h-1},u_{2h}: u \in S, h\in [r+1] \}\cup  \{ \basevertex{e}{},\basevertex{\tilde{e}}{} : e \in E(G[S]) \} \\
  \hspace{1.5cm}\bigcup \hspace{.5em}\cup_{ i\in [k]} \{p_{\lceil \nicefrac{\ell}{2} \rceil}^i, \tilde{p}_{\lceil \nicefrac{\ell}{2} \rceil}^i, a_{j, \roof{\nicefrac{\ell}{2^j}}}^i,   \tilde{a}_{j, \roof{\nicefrac{\ell}{2^j}}}^i,
 b_{j, \lceil \nicefrac{\ell}{2^{j+1}}\rceil}^i,\tilde{b}_{j, \lceil \nicefrac{\ell}{2^{j+1}}\rceil}^i: \\  \hspace{4cm}\text{ $S$ contains the $\ell^{th}$ vertex of $V_{i}$ }, j\in[\log_2 (\nicefrac{n}{2})]\}\\
\hspace{1.5cm} \bigcup\hspace{.5em} \cup_{\{i,j\}\subseteq [k], i<j}\{\specialvertex{q}{\roof{\ell/2}}{ij}, \specialvertex{\tilde{q}}{\roof{\ell/2}}{ij}, 
  \specialvertex{c}{h, \roof{\ell/2^h}}{ij},  \specialvertex{\tilde{c}}{h, \roof{\ell/2^h}}{ij},  \specialvertex{d}{h,\roof{\ell/2^{h+1}}}{ij},   \specialvertex{\tilde{d}}{h, \roof{\ell/2^{h+1}}}{ij} :\\
  \hspace{4cm} \text{ $G[S]$ contains the $\ell^{th}$ edge of $E_{ij}$, } h\in[\log(m/2)] \} \end{multline*}


\begin{clm}\label{no bp on u3,u4}For any  $u \in  S$ and any $h\in [r]$, there is no blocking edge with respect to $\eta$ that is incident to the vertex $u_{2h+1}$ or $u_{2h+2}$. \hide{\ma{Done}}
 \end{clm}
 \begin{proof} 
 For any value $h\in [r]$, vertex $\basevertex{u}{2h+1}$ is matched to its most preferred vertex in $\eta$, namely $\basevertex{u}{2h}$. Therefore, there is no blocking edge incident on $\basevertex{u}{2h+1}$.
 For any $h' \in [r-1]$, we have $N(\basevertex{u}{2h'+2}) =\pref{\basevertex{u}{2h'+1},\basevertex{u}{2h'+3}}$. Thus, there is no  blocking edge incident to $\basevertex{u}{2h'+2}$. 
 
Suppose that $u$ is the $\ell^{th}$ vertex in $V_{i}$. Then, we have $N(\basevertex{u}{2r+2})=\pref{\basevertex{u}{2r+1}, \tilde{p}_{ \lceil \nicefrac{\ell}{2}\rceil}^i}$, and we know that the edge $\basevertex{u}{2r+2}\tilde{p}_{ \lceil \nicefrac{\ell}{2}\rceil}^i$ is in $\eta$. However, since $\basevertex{u}{2r+1}$ is matched to its most preferred neighbor, it follows that there is no blocking edge incident to $\basevertex{u}{2r+2}$. 
\end{proof}

 \begin{clm}\label{u1u2 bp}For any vertex $u\in S$, $\basevertex{u}{1} \basevertex{u}{2}$ is a blocking edge in $G'$ with respect to $\eta$. Moreover, there is no other blocking edge incident to \basevertex{u}{1}  or \basevertex{u}{2} in $G'$\hide{\ma{Done}}
 \end{clm}

 \begin{proof} Since vertices \basevertex{u}{1} and \basevertex{u}{2} in $G'$ prefer each other over any other vertex, and the edge $\basevertex{u}{1} \basevertex{u}{2}$ is not in $\eta$, it must be a blocking edge with respect to $\eta$. 
 
Let $\ell = \sigma(V_{i},u)$, i.e the solution $S$ contains the $\ell^{th}$ vertex of the set $V_{i}$. Then, $N(u_{1})=\pref{\basevertex{u}{2},p_{\lceil \nicefrac{\ell}{2}\rceil}^i}$, and we know that $u_{1}p_{\lceil \nicefrac{\ell}{2}\rceil}^i \in \eta$. Thus, other than \basevertex{u}{1} \basevertex{u}{2}, there is no other blocking edge incident to \basevertex{u}{1} in $\eta$. Similarly, since $N(u_{ 2})=\pref{\basevertex{u}{1},\basevertex{u}{3}}$, and $\basevertex{u}{2} \basevertex{u}{3} \in \eta$, it follows that there is no other blocking edge incident to $\basevertex{u}{2}$. 
 \end{proof}
 
  \begin{clm}\label{no bp on p}
  For any $i\in [k]$ and $\ell \in [\nicefrac{n}{2}]$, there is no blocking edge with respect to $\eta$ that is incident to vertex $p_{\ell}^i$ or $\tilde{p}_{\ell}^i$. \hide{\ma{Done}}
 \end{clm}
 
 \begin{proof}
 Let $u= V_{i}(2\ell -1)$ and $v=V_{i}(2\ell)$, i.e, $u$ and $v$ denote the $2\ell-1^{st}$ and $2\ell^{th}$ elements of $V_{i}$, respectively. 
 
Suppose that $\{u, v\} \cap S = \emptyset$. Then, due to the construction of $\eta$, we know that 
 $\basevertex{u}{1} \basevertex{u}{2}$ and $\basevertex{v}{1} \basevertex{v}{2}$ are in $\eta$. Recall that $N(\specialvertex{p}{\ell}{i})=\pref{\basevertex{u}{1}, \basevertex{v}{1}, \specialvertex{a}{1\ell}{i}}$. Since  \basevertex{u}{1} and \basevertex{v}{1} are matched to their most preferred neighbor in $\eta$,  namely  \basevertex{u}{2} and \basevertex{v}{2}, so there is no blocking edge incident to \specialvertex{p}{\ell}{i}. Hence, this case is resolved.
 
 Suppose that $u\in S$. Then, $\specialvertex{p}{\ell}{i} \basevertex{u}{1} \in \eta$. Since \specialvertex{p}{\ell}{i} prefers  \basevertex{u}{1} over any other vertex, so there is no blocking edge incident to  \specialvertex{p}{\ell}{i}. 
 
Suppose that $v\in S$. Then, $\specialvertex{p}{\ell}{i} \basevertex{v}{1} \in \eta$, by the construction of $\eta$. Since $|S \cap V_{i}|=1$, and $v\in S$, it follows that $u\notin S$. Hence, $\basevertex{u}{1}\basevertex{u}{2} \in \eta$, implying that \basevertex{u}{1} is matched to its most preferred neighbor \basevertex{u}{2}. Therefore, there is no blocking edge \wrt $\eta$ that is incident to $\specialvertex{p}{\ell}{i}$. By symmetry, we can argue that  
 there  is no blocking edge incident to $\specialvertex{\tilde{p}}{\ell}{i}$. 
  \end{proof}

 



  \begin{clm}\label{no bp on a}For any $i\in [k], j\in [\log_2(\nicefrac{n}{2})]$ and $\ell \in [\beta_j]$, there is no blocking edge with respect to $\eta$ that is incident to vertex $a_{j, \ell}^i$ or $\tilde{a}_{j, \ell}^i$. 
 \end{clm}
 
 \begin{proof}
 We first consider the case when $j=1$. Recall that $N(a_{1,\ell}^i)=\pref{p_{\ell}^i,b_{1,\ell}^i}$. If $a_{1,\ell}^i p_{\ell}^i \in \eta$, then $\specialvertex{a}{1,\ell}{i}$ is matched to its most preferred vertex. Thus, there is no blocking edge incident on  $\specialvertex{a}{1,\ell}{i}$. Suppose that $a_{1,\ell}^i b_{1,\ell}^i \in \eta$. In this case, by the construction of $\eta$, either $p_{\ell}^i \basevertex{u}{1} \in \eta$ or $\specialvertex{p}{\ell}{i} \basevertex{v}{1}\in \eta$, where $u=V_{i}(2\ell-1)$ and $v=V_{i}(2\ell)$.  Note that $\specialvertex{p}{\ell}{i}$ prefers both  \basevertex{u}{1} and \basevertex{v}{1} over $a_{1,\ell}^i$. Hence, there is no blocking edge incident to $a_{1,\ell}^i$. 
 Next, we consider the case when $j \geq 2$. Recall that $N(a_{j,\ell}^i)=\pref{\specialvertex{b}{j-1,\ell}{i}, \specialvertex{b}{j,\ell}{i} }$. If $\specialvertex{a}{j, \ell}{i} \specialvertex{b}{j-1, \ell}{i} \in \eta$, then  $a_{j, \ell}^i$ is matched to its most preferred vertex. Thus, there is no blocking edge incident on  $\specialvertex{a}{j\ell}{i}$. Suppose that $\specialvertex{a}{j, \ell}{i} \specialvertex{b}{j, \ell}{i} \in \eta$. Since $\specialvertex{a}{j, \ell}{i}$ is the last preference of $b_{j-1, \ell}^i$ (and $b_{j-1, \ell}^i$ is matched to $\specialvertex{a}{j-1, \ell}{i}$ in $\eta$), we can conclude that there is no  blocking edge incident to  $\specialvertex{a}{j,\ell}{i}$. Similarly, there is no blocking edge with respect to $\eta$ that is incident to $\tilde{a}_{j,\ell}^i$  \end{proof}
 
 
 \begin{clm}\label{no bp on b}For any $i\in [k], j\in [\log_2(\nicefrac{n}{2})]$, and $\ell \in [\gamma_j]$, there is no blocking edge with respect to $\eta$ that is incident to vertex $b_{j\ell}^i$ or $\tilde{b}_{j\ell}^i$. \hide{\ma{NONE}}
 \end{clm}

 \begin{proof}
 We first consider the case when $j\in [\log_2(\nicefrac{n}{2}) -1]$. Recall that\\ $N(b_{j,\ell}^i)= \pref{a_{j, 2\ell -1}^i,~a_{j, 2\ell}^i, ~a_{j+1, \ell}^i }$. If $b_{j, \ell}^ia_{j, 2\ell -1}^i \in \eta$, then $b_{j, \ell}^i$ is matched to its most preferred vertex. Hence, there is no blocking edge incident on $b_{j,\ell}^i$. Suppose that $b_{j, \ell}^ia_{j ,2\ell}^i \in \eta$. Note that $a_{j, 2\ell -1}^i$ is the only vertex that $b_{j, \ell}^i$ prefers over $a_{j, 2\ell}^i$. Since $b_{j, \ell}^i$ is the last preference of $a_{j, 2\ell -1}^i$, and $a_{j, 2\ell -1}^i$ is saturated in $\eta$ (Claims~\ref{tilde eta is perfect} and \ref{size-eta} imply  that $\eta$ is a perfect matching), there is no blocking edge incident on $b_{j, \ell}^i$. If $b_{j, \ell}^ia_{j+1, \ell}^i \in \eta$, then using the same argument as earlier, there is no blocking edge incident on $b_{j, \ell}^i$. Now, consider the case when $j=\log_2(\nicefrac{n}{2})$. Since $N(b_{j,1}^i)= \pref{a_{j,1}^i,a_{j,2}^i }$, there is no blocking edge incident on $b_{j,1}^i$ using the same arguments as earlier. Similarly, there is no blocking edge  incident on $\tilde{b}_{j, \ell}^i$ with respect to $\eta$. 
 \end{proof}
 
\begin{clm}\label{no bp on q}For any $\{i,j\} \subseteq [k], i<j$, and $\ell \in [\nicefrac{m}{2}]$, there is no blocking edge with respect to $\eta$ that is incident to vertex $q_{\ell}^{ij}$ or $\tilde{q}_{\ell}^{ij}$. 
 \end{clm}
 \begin{proof}
 The proof is similar to the proof of Claim~\ref{no bp on p}.
 \end{proof}
 \begin{clm}\label{no bp on c}For any $\{i,j\} \subseteq [k], i<j$,  $h\in [\log_2(\nicefrac{m}{2})]$ and $\ell \in [\rho_h]$, there is no blocking edge with respect to $\eta$ that is incident to vertex $\specialvertex{c}{h,\ell}{ij}$ or $\specialvertex{\tilde{c}}{h,\ell}{ij}$. 
 \end{clm}
 \begin{proof}
  The proof is similar to the proof of Claim~\ref{no bp on a}.
 \end{proof}
 \begin{clm}\label{no bp on d}For any $\{i,j\} \subseteq [k], i<j, h\in [\log_2(\nicefrac{m}{2})], \ell \in [\tau_h]$, there is no blocking edge with respect to $\eta$ that is incident to vertex $\specialvertex{d}{h,\ell}{ij}$ or $\specialvertex{\tilde{d}}{h,\ell}{ij}$.  
  \end{clm}
 \begin{proof}
  The proof is similar to the proof of Claim~\ref{no bp on b}.
  \end{proof}

 \begin{clm}\label{bp on e}Let $e$ denote an edge in the clique $G[S]$.
 Then, the edge $\basevertex{e}{} \basevertex{\tilde{e}}{}$ in $G'$ is a blocking edge with respect to $\eta$. Moreover, there is no other blocking edge with respect to $\eta$ that is incident to vertex $\basevertex{e}{}$ or $\basevertex{\tilde{e}}{}$ in $G'$. 
 
 \end{clm}
 
   \begin{proof}
 Since vertices $\basevertex{e}{}$ and $\basevertex{\tilde{e}}{}$ prefer each other over any other vertex in $G'$, and the edge $\basevertex{e}{}\basevertex{\tilde{e}}{}$ is not in $\eta$, it must be a blocking edge with respect to $\eta$. 
 
 Let $e=uv$, that is vertices $u \in V_{i}$ and $v \in V_{j}$ are the two endpoints of the edge $e$ in $G$. Suppose that for some $h, h'\in [r]$, we have $h=\sigma(E_{u}, e)$ and $h'=\sigma(E_{v}, e)$ i.e, $e$ is the $h^{th}$ element of $E_{u}$ and the $h'^{th}$ element of $E_{v}$. 
 
 Recall that $N_{G'}(\basevertex{e}{ })=\pref{\basevertex{\tilde{e}}{}, \basevertex{u}{2h+1}, \basevertex{v}{2h'+1}, q_{\lceil \nicefrac{\ell}{2}\rceil}^{ij}}$, where $\ell=\sigma(E_{ij},e)$, i.e, $e$ is the $\ell^{th}$ element in the set $E_{ij}$. By the construction of $\eta$, we know that the edge $\basevertex{e}{} q_{\lceil \nicefrac{\ell}{2}\rceil}^{ij}$ is in $\eta$. Moreover, since $\{u, v\} \subseteq S$, we know that vertices $\basevertex{u}{2h+1}$ and $\basevertex{v}{2h'+1}$ are matched to their most preferred vertices in $\eta$.
Hence, $\basevertex{e}{} \basevertex{\tilde{e}}{}$ must be the only blocking edge with respect to $\eta$ that is incident to $\basevertex{e}{}$. Similarly, we note that since $N(\basevertex{\tilde{e}}{})=\pref{\basevertex{e}{}, \tilde{q}_{\lceil \nicefrac{\ell}{2}\rceil}^{ij}}$ and edge $\basevertex{\tilde{e}}{} \tilde{q}_{\lceil \nicefrac{\ell}{2}\rceil}^{ij}$ is in $\eta$, the only blocking edge that is incident to the vertex $\basevertex{\tilde{e}}{}$ is $\basevertex{e}{}\basevertex{\tilde{e}}{}$. Thus, the claim is proved.
   \end{proof}
   
 

Note that Claims~\ref{no bp on u3,u4} and \ref{u1u2 bp} imply that for each vertex $u\in S$, there is a unique blocking edge \wrt $\eta$ (namely $\basevertex{u}{1} \basevertex{u}{2}$); and Claim~\ref{bp on e} implies that for each edge $e$ in $G[S]$, there is a unique blocking edge (namely $\basevertex{e}{} \basevertex{\tilde{e}}{}$ )\wrt $\eta$. Moreover, Claims~\ref{no bp on p}--\ref{no bp on d} imply that there are no other blocking edges \wrt $\eta$. Hence, in total there are $k'=k+\nicefrac{k(k-1)}{2}$ blocking edges \wrt $\eta$. Thus, we can conclude that the forward direction is proved $(\Rightarrow$). 

  
 $(\Leftarrow)$ In the reverse direction, let $\eta$ be a matching of size at least $\lvert \mu \rvert+k+\nicefrac{k(k-1)}{2}$ such that $\eta$ has at most $k+\nicefrac{k(k-1)}{2}$ blocking edges. Due to the size of $\eta$, we can infer that it is a perfect matching.
 
 Let ${\mathtt B}_{\eta}$ be the set of blocking edges with respect to $\eta$. We first note some properties of matching  $\eta$ and the set ${\mathtt B}_{\eta}$. We start by identifying the edges in ${\mathtt B}_{\eta}$.



Note that in our instance, the static edges in $G'$ are of the following type: For any $u\in V(G)$, edge 
$\basevertex{u}{1} \basevertex{u}{2}$ in $G'$ is a static edge and is called the {\it u-type static edge}; for any $e\in E(G)$, edge $\basevertex{e}{ } \basevertex{\tilde{e}}{}$ in $G'$ is a static edge and is called the {\it e-type static edge}. 


In the following claims, we prove that a blocking edge with respect to $\eta$ is  either a u-type static edge or e-type static edge. In fact, for each $i\in [k]$, there is unique u-type static edge which is a blocking edge, and for each $\{i,j\}\subseteq [k]$, there is unique e-type static edge which is a blocking edge.  


\begin{clm}[{\bf u-type static edge}]
\label{claim:solution intersect each part}
For each $i\in [k]$, there exists $u\in V_{i}$,  such that $u_{1}u_{2}$ is a blocking edge \wrt $\eta$. 
\end{clm}   
\begin{proof}
Since $\eta$ is a perfect matching, for each $i\in [k]$ and $j=\log_2(\nicefrac{n}{2})$, vertex $b_{j,1}^i$ is saturated by $\eta$. Recall that $N(\specialvertex{b}{j, 1}{i})=\pref{\specialvertex{a}{j,1}{i}, \specialvertex{a}{j,2}{i}}$. Therefore, there exists a (unique) $z\in [2]$, such that $\specialvertex{b}{j, 1}{i} \specialvertex{a}{j, z}{i} \in \eta$. 


Since $\eta$ is a perfect matching and \specialvertex{b}{j-1, z}{i} has two other neighbors \specialvertex{a}{j-1, 2z-1}{i} and \specialvertex{a}{j-1,2z}{i}, it follows that either $\specialvertex{b}{j-1, z}{i} \specialvertex{a}{j-1, 2z-1}{i} \in \eta$ or $\specialvertex{b}{j-1, z}{i} \specialvertex{a}{j-1, 2z}{i} \in \eta$. We view the index $j$ as indicating a level, the highest being $\log_2(\nicefrac{n}{2})$. As we go down each level starting from the highest, we obtain a matching edge in $\eta$. The lowest level is reached when for some value $h\in [\nicefrac{n}{4}]$, we reach the vertex $\specialvertex{b}{1, h}{i}$. For this vertex, there are two possible matching partners in $\eta$: $\specialvertex{a}{1, 2h-1}{i} $ or $\specialvertex{a}{1, 2h}{i}$. Thus, for some value $h' \in \{2h-1, 2h\}$,  edge $\specialvertex{b}{1, h}{i}\specialvertex{a}{1, h'}{i} \in \eta$.


Since $\eta$ is a perfect matching, \specialvertex{p}{h'}{i} must be matched to either $\basevertex{x}{1}$ or $\basevertex{y}{1}$ (its other two neighbors) in $\eta$, where $x=V_{i}(2h'-1)$ and $y=V_{i}(2h')$ i.e, $x$ is the $2h'-1^{st}$ element of $V_{i}$ and $y$ is the $2h'^{th}$ element of $V_{i}$. If $\specialvertex{p}{h'}{i} \basevertex{x}{1} \in \eta$, then since $\basevertex{x}{1}$ and $\basevertex{x}{2}$ are each others first preference, the edge $\basevertex{x}{1}\basevertex{x}{2} \in {\mathtt B}_{\eta}$. Otherwise, if $\specialvertex{p}{h'}{i} \basevertex{y}{1} \in \eta$, then with analogous argument, it follows that the edge $\basevertex{y}{1}\basevertex{y}{2} \in {\mathtt B}_{\eta}$. Hence, the result is proved.
\end{proof}


\begin{clm}[{\bf e-type static edge}]
\label{claim:solution intersect clique edge}
For each $\{i,j\}\subseteq [k]$, there exists $e \in E_{ij}$, such that $\basevertex{e}{}\basevertex{\tilde{e}}{}$ is a blocking edge \wrt $\eta$ and $q_{\lceil \nicefrac{\ell}{2} \rceil}^{ij}\basevertex{e}{}\in \eta$, where $e$ is the $\ell^{th}$ element of $E_{ij}$.
\end{clm}   
\begin{proof}
Since $\eta$ is a perfect matching, for each $\{i,j\}\subseteq [k]$ and $h=\log_2(\nicefrac{m}{2})$, 
the vertex $\specialvertex{d}{h,1}{ij}$ must be saturated by $\eta$. Recall that $N(\specialvertex{d}{h, 1}{ij})=\pref{\specialvertex{c}{h,1}{ij}, \specialvertex{c}{h,2}{ij}}$. Therefore, there exists a (unique) $z\in [2]$, such that $\specialvertex{d}{h, 1}{ij} \specialvertex{c}{h, z}{ij} \in \eta$. 

Now, since $\eta$ is a perfect matching, and $N(\specialvertex{d}{h-1, z}{ij})=\pref{\specialvertex{c}{h-1, 2z-1}{ij},\specialvertex{c}{h-1, 2z}{ij}, \specialvertex{c}{h, z}{ij}}$, either $\specialvertex{d}{h-1, z}{ij}\specialvertex{c}{h-1, 2z-1}{ij} \in \eta$ or $\specialvertex{d}{h-1, z}{ij} \specialvertex{c}{h-1, 2z}{ij} \in \eta$. We view the index $h$ as indicating a level, the highest being $\log_2(\nicefrac{m}{2})$. As we go down each level starting from the highest, we obtain a matching edge in $\eta$. The lowest level is reached when for some value $h'\in [\nicefrac{m}{4}]$, we reach the vertex $\specialvertex{d}{1, h'}{ij}$. For this vertex, there are two possible matching partners in $\eta$: $\specialvertex{c}{1, 2h'-1}{i j} $ or $\specialvertex{c}{1, 2h'}{ij}$. Thus, for some $\bar{h}\in \{2h' -1, 2h'\}$, edge $\specialvertex{d}{1,h'}{ij} \specialvertex{c}{1, \bar{h}}{ij} \in \eta$. 

Since $\eta$ is a perfect matching, $ \specialvertex{q}{\bar{h}}{ij}$ is matched to either $\basevertex{e}{}$ or $\basevertex{e'}{}$ in $\eta$, where $e=E_{ij}(2\bar{h}-1)$ and $e'=E_{ij}(2\bar{h})$, i.e, $e$ is the $2\bar{h}-1^{st}$ element of $E_{ij}$ and $e'$ is the $2\bar{h}^{th}$ element of $E_{ij}$. If $\specialvertex{q}{\bar{h}}{ij} \basevertex{e}{} \in \eta$, then since $\basevertex{e}{}$ and $\basevertex{\tilde{e}}{}$ are each others first preference, edge $\basevertex{e}{}\basevertex{\tilde{e}}{}\in {\mathtt B}_{\eta}$. Else if $\specialvertex{q}{\bar{h}}{ij} \basevertex{e'}{} \in \eta$, then with analogous argument, it follows that $\basevertex{e'}{} \basevertex{\tilde{e'}}{} \in {\mathtt B}_{\eta}$. Hence, the result is proved.
\end{proof}



%

\begin{corollary}\label{obs:size B}
 For each $i\in [k]$, there exists a unique $u\in V_{i}$,  such that the edge $\basevertex{u}{1} \basevertex{u}{2}$ is a blocking edge \wrt $\eta$; and for each $\{i,j\}\subseteq [k]$, there exists a unique $e \in E_{ij}$, such that $\basevertex{e}{} \basevertex{\tilde{e}}{}$ is a blocking edge \wrt $\eta$. 
\end{corollary}

\begin{proof}
Using Claims ~\ref{claim:solution intersect each part} and ~\ref{claim:solution intersect clique edge}, 
we know that there are at least $k+\nicefrac{k(k-1)}{2}$ blocking edges \wrt $\eta$. Since $k' =k+\nicefrac{k(k-1)}{2}$, the uniqueness condition follows.
\end{proof}

Conversely, we can also argue the following.

\begin{corollary}\label{obs:type of bp}
Any blocking edge \wrt $\eta$ is either a $u$-type static edge or an $e$-type static edge.  
\end{corollary}

\begin{proof}
Using Corollary~\ref{obs:size B}, we know that there are at least $k$ u-type blocking edges and $\nicefrac{k(k-1)}{2}$ $e$-type blocking edges \wrt $\eta$. Since $k' =k+\nicefrac{k(k-1)}{2}$, there cannot exist any other (besides $u$-type and $e$-type) blocking edge \wrt $\eta$. 
\end{proof}

Next, we prove that the $e$-type (static) blocking edges force certain edges to be in the matching $\eta$. 

\begin{clm}\label{q1 pair}For any $\{i,j\}\subseteq [k]$, consider some $e\in E_{ij}$ such that  $\basevertex{e}{}\basevertex{\tilde{e}}{}$ is a blocking edge \wrt $\eta$. Then, for the value $\ell = \sigma(E_{ij}, e)$, the edge $\specialvertex{q}{\roof{ \nicefrac{\ell}{2}}}{ij} \basevertex{e}{}$ is in $\eta$. 


\end{clm}

\begin{proof} By Claim~\ref{claim:solution intersect clique edge}, there exists an edge $e^{\prime}\in E_{ij}$ such that $\basevertex{e'}{}\basevertex{\tilde{e'}}{}$ is a blocking edge \wrt $\eta$ and $\specialvertex{q}{ \roof{ \nicefrac{\ell'}{2}}}{ij} \basevertex{e'}{}$ is in $\eta$, where $\ell'=\sigma(E_{ij},e')$. By Corollary~\ref{obs:size B}, we know that $e^{\prime}=e$. 
\end{proof}


\begin{sloppypar}
\begin{clm}[{\bf consistency between u-type static edge and e-type static edge}]\label{consistency} 
Suppose that for some $\{i,j\}\subseteq [k], i<j$, we have $e\in E_{ij}$ such that $\basevertex{e}{}\basevertex{\tilde{e}}{}$ is a blocking edge \wrt $\eta$. Let $u $ and $v$ denote the two endpoints of the edge $e$ in $G$. Then, both 
 $\basevertex{u}{1}\basevertex{u}{2}$ and $\basevertex{v}{1}\basevertex{v}{2}$ are blocking edges \wrt $\eta$. \hide{\ma{DOne}}

\end{clm}
\end{sloppypar}
\begin{proof}

For the sake of contradiction, suppose that both $\basevertex{u}{1} \basevertex{u}{2}$ and $\basevertex{v}{1} \basevertex{v}{2}$ are not blocking edges \wrt $\eta$. Without loss of generality, we may assume that $\basevertex{u}{1} \basevertex{u}{2}$ is not a blocking edge. Since $u_1$ and $u_2$ prefer each other over any other vertex, $u_1u_2 \in \eta$, otherwise it will contradict the fact that $u_1u_2$ is not a blocking edge. For any $h\in [r]$, suppose that $u_{2h+1}e' \in \eta$, where $e'=\sigma(E_u,h)$. Since $e'$ and $\tilde{e'}$ prefer each other over any other vertex $e'\tilde{e'}\in \mathtt{B}_\eta$. Since $u_{2h+1}e' \in \eta$, the edge $\specialvertex{q}{\roof{\nicefrac{\ell}{2}}}{ij} \basevertex{e'}{} \notin \eta$, where $\ell = \sigma(E_{ij}, e')$, a contradiction to Claim~\ref{q1 pair}. Thus, for any $h\in [r]$,  $u_{2h+1}e' \notin \eta$, where $e'=\sigma(E_u,h)$. Since $u_1u_2 \in \eta$ and $\eta$ is a perfect matching, we can infer that for each $h\in [r]$, $u_{2h+1}u_{2h+2}\in \eta$. Since $e\tilde{e} \in \mathtt{B}_\eta$, due to  Claim~\ref{q1 pair}, $\specialvertex{q}{\roof{\nicefrac{\ell}{2}}}{ij} \basevertex{e}{} \in \eta$, where $\ell = \sigma(E_{ij}, e)$. Note that there exists $h\in [r]$ such that $u_{2h+1}e \in E(G')$. Since $u_{2h+1}$ prefers $e$ more than its matched partner in $\eta$, i.e., $u_{2h+2}$, and $e$ prefers $u_{2h+1}$ more than its matched partner in $\eta$, $u_{2h+1}e \in \mathtt{B}_\eta$, a contradiction to Corollary~\ref{obs:type of bp}. 

\end{proof}

\begin{supress}
\il{OLD proof of Claim~\ref{consistency}: I don't understand}
Suppose that $e_{\ell 1}^{ij}e_{\ell 2}^{ij} \in {\mathtt B}$ but $\{u_{z1}^iu_{z2}^i,u_{z'1}^ju_{z'2}^j\}\nsubseteq {\mathtt B}$. Without loss of generality, let $u_{z1}^iu_{z2}^i \notin {\mathtt B}$. Let $e_{\ell 1}^{ij}=E_{u_z^i}(h)$, where $h\in [r]$. Since $u_{z1}^i$ and $u_{z2}^i$ prefer each other over any other vertex, and $u_{z1}^iu_{z2}^i \notin {\mathtt B}$, it follows that $u_{z1}^iu_{z2}^i \in \eta$. 

Recall that for each $\hat{i}\in [k]$, $\hat{\ell} \in [n]$,  $E_{u_{\hat{\ell}}^{\hat{i}}}= \{e_{s1}^{\hat{i}\hat{j}}\mid \hat{j}\in [k], \hat{i}<\hat{j}, \text{ and } e_s^{{\hat{i}\hat{j}}}(\in E_{\hat{i}\hat{j}}) \text{ is incident on } u_{\hat{\ell}}^{\hat{i}} \in V_{\hat{i}}\}$ is an ordered set. 

\Ma{Suppose that $E_{u_z^i}(h')u_{z(2h'+1)}^{i} \in \eta$,}\ma{PJ, this cant be $E_{u_z^i}(h')$ is an edge!} where $h'\in [r]$. Let  $e_{\ell'1}^{ij'}=E_{u_z^i}(h')$, where $\ell' \in [m], j'\in [k]$. Since $e_{\ell'1}^{ij'}$ and $e_{\ell'2}^{ij'}$  prefer each other over any other vertex, $e_{\ell'1}^{ij'}e_{\ell'2}^{ij'} \in {\mathtt B}$. Since $e_{\ell'1}^{ij'}u_{z(2h'+1)}^{i} \in \eta$, $e_{\ell'1}^{ij'}q_{ \lceil \nicefrac{\ell'}{2} \rceil}^{ij'} \notin \eta$, a contradiction to Claim~\ref{q1 pair}. Thus, for any $h'\in [r]$, $E_{u_z^i}(h')u_{z(2h'+1)}^{i} \notin \eta$. Since $u_{z1}^iu_{z2}^i \in \eta$, and $\eta$ is a perfect matching, we can infer that $u_{z(2h'+1)}^{i}u_{z(2h'+2)}^{i} \in \eta$, for all $h'\in [r]$. Since $e_{\ell 1}^{ij}e_{\ell 2}^{ij} \in {\mathtt B}$,  $e_{\ell1}^{ij}q_{\lceil \nicefrac{\ell}{2} \rceil}^{ij} \in \eta$ due to  Claim~\ref{q1 pair}. Since $u_{z(2h+1)}^{i}$ prefers $e_{\ell 1}^{ij}$ over $u_{z(2h+2)}^{i}$($= \eta(u_{z(2h+1)}^{i})$), and $e_{\ell 1}^{ij}$ prefers $u_{z(2h+1)}^{i}$ over $q_{ \lceil \nicefrac{\ell}{2} \rceil}^{ij}$($= \eta(e_{\ell 1}^{ij})$), it follows that $u_{z(2h+1)}^{i}e_{\ell 1}^{ij} \in {\mathtt B}$, a contradiction to Corollary~\ref{obs:type of bp}.  

\il{}

\end{supress}

Next, we construct two sets $S$ and $E_S$ as follows. Let $S=\{u\in V(G):  
\basevertex{u}{1}\basevertex{u}{2}\in {\mathtt B}_{\eta}\}$, i.e, the set of vertices in $G$ that correspond to a $u$-type static blocking edge. Let $E_S=\{e\in E(G): 
\basevertex{e}{}\basevertex{\tilde{e}}{} \in  {\mathtt B}_{\eta}\}$, i.e, the set of edges in $G$ that correspond to a $e$-type static blocking edge.

We claim that $G_S=(S,E_S)$ is a clique, and $|S\cap V_i|=1$, for each $i\in [k]$. Using Claim \ref{consistency}, we know that for each edge $e \in E_{S}$, we have $\{u, v\} \sse S$, where $u$ and $v$ are the two endpoints of the edge $e$. 

Moreover, using Corollary~\ref{obs:size B}, we that $|V_{i}\cap S|=1$ for each $i\in [k]$ and $|E_S|=\nicefrac{k(k-1)}{2}$. Hence, we may conclude that $G_S$ is a clique on $k$ vertices. This completes the proof of the lemma.
\end{proof}

Thus, Theorem~\ref{thm:asm} is proved.




\section{\W$[1]$-hardness of \ams}
In this section, we show the parameterized intractability of \ams with respect to several parameters. In particular, we prove Theorem~\ref{thm:whard-k-and-t} and Theorem~\ref{thm:smd w-hard}. 

\subsection{Proof of Theorem~\ref{thm:whard-k-and-t}}

We again give a polynomial-time parameter preserving many-to-one reduction from \mcqsmall on regular graphs. Let $(G,k)$ be an instance of \mcqsmall. To construct an instance $(G',\Co{L},\mu,k',q,t)$ of \ams, we construct a graph $G'$, a set of $\Co{L}$ containing the preference list of each vertex of $G'$, and a stable matching $\mu$ as defined in the proof of Theorem~\ref{thm:asm}. 
We set the parameters $k'$ and $t$ also as in the proof of Theorem 1. 
We set parameter $q$ as follows:
$$ q= (2r+3)k+\frac{3k(k-1)}{2}+4k\log_2 \Big(\frac{n}{2}\Big)+2k(k-1) \log_2 \Big(\frac{m}{2}\Big)$$ Next, we show that $(G,k)$ is a \yes-instance of \mcqsmall if and only if $(G',\pref,\mu,k',q,t)$ is a \yes-instance of \ams. 
In the forward direction, let $X$ be a solution of \mcqsmall for $(G,k)$. We construct a matching $\eta$ as defined in the above proof. As proved above, $|\eta|=|\mu|+t$ and the number of blocking edges with respect to $\eta$ is $k'$. Now, we show that $|\mu \triangle \eta|\leq q$. Recall that for each vertex in $X$, we delete $r+2\log_2(\nicefrac{n}{2})+1$ edges from $\eta$ (which also belongs to $\mu$), and add $r+2\log_2(\nicefrac{n}{2})+2$ edges to $\eta$. Similarly, for each edge in $E(G[X])$, we delete $2 (\log_2 \nicefrac{m}{2}) +1$ edge from $\eta$ which is also in $\mu$, and add  $2 (\log_2 \nicefrac{m}{2}) +2$ edges to $\eta$. Hence, 
$$|\mu \triangle \eta|= (2r+3)k+\frac{3k(k-1)}{2}+4k\log_2 \Big(\frac{n}{2}\Big)+2k(k-1) \log_2 \Big(\frac{m}{2}\Big)$$ This completes the proof in the forward direction. The proof of backward direction is same as the proof of the backward direction of Theorem~\ref{thm:asm}. 

\subsection{Proof of Theorem~\ref{thm:smd w-hard}}

\newcommand{\ordering}[1]{\ensuremath{ [#1] }}


We again give a polynomial-time parameter preserving many-to-one reduction from \mcqsmall similar to the one in Theorem~\ref{thm:asm}. Here, we do not need graph to be a regular graph. 
\par

\noindent{\bf Construction.} Given an instance $\Co{I}=(G,(V_1,\ldots,V_k))$ of \mcqsmall, we construct an instance $\Co{J}=(G',\Co{L},\mu,k',q,t)$ of \ams as follows.  
 For any $\{i, j\} \sse [k]$, such that $i<j$, we use $E_{ij}$ to denote the set of edges between sets $V_i$ and $V_j$. 
\begin{sloppypar}
\begin{itemize}[noitemsep]
\item For each vertex $v \in V(G)$, we add four vertices in $G'$, denoted by $\{\basevertex{u}{i} : i\in [4]\}$, connected via a path: $(\basevertex{u}{1},\basevertex{u}{2},\basevertex{u}{3},\basevertex{u}{4})$ in $G'$. 
\item  For each edge $e \in E_{ij}$, we add vertices $\basevertex{e}{}$ and $\basevertex{\tilde{e}}{}$ to $V(G')$, and the edge $\basevertex{e}{}\basevertex{\tilde{e}}{}$ to $E(G')$.  
\item For each $i\in[k]$, we add two vertices $p_1^i, p_2^i$, and for each $\{i,j\} \subseteq [k]$ where $i<j$,  we add two vertices $q_1^{ij}, q_2^{ij}$ to $V(G')$. 
\item 
For each $i\in[k]$ and for each vertex $u \in V_i$, we add two edges $\basevertex{u}{1}p_1^i$ and $\basevertex{u}{4}p_2^i$ to $E(G')$. 
 For each $\{i,j\} \subseteq [k]$, $i < j$, and for each edge $e \in E_{ij}$,
   we add four edges $q_1^{ij}\basevertex{e}{}, q_2^{ij}\basevertex{\tilde{e}}{}, \basevertex{e}{}\basevertex{u}{3}$, and $\basevertex{e}{}\basevertex{v}{3}$ to $E(G')$.  
 \end{itemize}
\end{sloppypar}
Figure \ref{fig:hardness} describes the construction of $G'$. 
Note that $V(G')=4\lvert V(G)\rvert+2\lvert E(G)\rvert+2k+k(k-1)$. Recall that in the construction of an instance in the proof of Theorem~\ref{thm:asm}, for each vertex in $V(G)$, we added a path of length $2r+2$, while here we add a path of length $4$. Moreover, instead of adding $n$ vertices $p_{\ell}^i$ and $\tilde{p}_\ell^i$, for each $i\in [k]$, $\ell\in [\nicefrac{n}{2}]$, we only add two vertices $p_1^i$ and $p_2^i$. Similarly, we added only two vertices $q_1^{ij}$ and $q_2^{ij}$ instead of adding $m$ such vertices. Furthermore, here we did not add the other special vertices which we added in the previous reduction. This is how we decrease the length of augmenting paths. But, note that degree of vertices $\basevertex{u}{3}$, $p_1^i, p_2^i,q_1^{ij},q_2^{ij}$, where $u\in V_i$, $\{i,j\}\subseteq [k], i<j$, is large.
 

 For any vertex $u\in V(G)$, we define \[\X{E}_{u}=\{ \basevertex{e}{}\in V(G') \colon e \in E(G) \text{ and } u \text{ is an endpoint of } e\}\]
The preference list of each vertex in $G'$ is presented in Table \ref{pref_list}.
\par
{\bf Matching $\mu$:} Let $\mu=\{\basevertex{u}{1}\basevertex{u}{2}, \basevertex{u}{3}\basevertex{u}{4}, \basevertex{e}{}\basevertex{\tilde{e}}{}\colon u \in V(G), e\in E(G), \text{ and } i \in [k]\}$. Clearly, $\mu$ is a matching.  Note that $\lvert \mu\rvert = 2\lvert V(G)\rvert +\lvert E(G)\rvert$.

\par
{\bf Parameter:} We set $k'=k+\nicefrac{k(k-1)}{2}$, $q=5k+\nicefrac{3k(k-1)}{2}$, and $t=k'$. 
\par
Clearly, this construction can be carried out in polynomial time. Next, we will prove some structural properties about our construction, namely that the graph $G'$ is bipartite (Claim~\ref{clm:nasm-bipartite-graph}) and $\mu$ is a stable matching (Claim~\ref{clm:nasm-stable-matching}). 



\begin{clm}\label{clm:nasm-bipartite-graph}
Graph $G'$ is bipartite. 
\end{clm}
\begin{proof}
 We show that $G'$ is a bipartite graph by creating a bipartition for $G'$ as follows. 
 For each $i \in [k]$, and each $u \in V_i$, we assign $p_1^i$, $\basevertex{u}{2}$ and $\basevertex{u}{4}$ to one part and $p_2^i$, $\basevertex{u}{1}$ and $\basevertex{u}{3}$ to another part. 
 For each $\{i,j\} \subseteq [k]$, $i<j$, since a vertex $\basevertex{e}{}\in V(G')$ (corresponding to the edge $e=uv, u \in V_i, v \in V_j$) is connected to $\basevertex{u}{3}$ and $\basevertex{v}{3}$, we assign $\basevertex{e}{}$ and $q_2^{ij}$ to the  part containing $p_1^i$, and assign $\basevertex{\tilde{e}}{}$ and $q_1^{ij}$ to the part containing $p_2^i$. Observe that each part is an independent set. Hence $G'$ is a bipartite graph. 
 \end{proof}
 
 \begin{clm}\label{clm:nasm-stable-matching}
 $\mu$ is a stable matching.
 \end{clm}
\begin{proof}
We begin by noting that for any vertex $u\in V(G)$, vertices $u_1$ and $u_2$ prefer each other over any other vertex in $G'$. Therefore, edge $u_1u_2$ is a static edge and must belong to every stable matching in $G'$. Similarly, for each $e \in E(G)$,  we note that $e\tilde{e}$ is a static edge in $G'$, and thus belongs to every stable matching in $G'$.  Since $\basevertex{u}{3}$ is the first preference of $\basevertex{u}{4}$, and the vertices which  $\basevertex{u}{3}$ prefers over $\basevertex{u}{4}$ ( i.e., $\basevertex{u}{2}$ and vertices in $\X{E}_{u}$) are matched to their first preferred vertices, it follows that there is no blocking edge with respect to $\mu$. Hence, $\mu$ is a stable matching in $G'$.
\end{proof}
\begin{table*}[t]
   \centering
 For each $i \in [k]$ and each $u \in V_i$, we have the following preference lists:
 \par
 \smallskip
	\begin{tabularx}{10em}{@{} X  X @{}}
    $\basevertex{u}{1} \colon$ & $\pref{\basevertex{u}{2},p_1^i}$ \\
    $\basevertex{u}{2} \colon$ & $\pref{\basevertex{u}{1},\basevertex{u}{3}}$ \\
    $\basevertex{u}{3} \colon$ & $\pref{\basevertex{u}{2},\ordering{\X{E}_{u}},\basevertex{u}{4}}$ \\
    $\basevertex{u}{4} \colon$ & $\pref{\basevertex{u}{3},p_2^i}$  \\
    \end{tabularx}
    \par
    \medskip
    For each edge $e\in E_{ij}$ with endpoints $u\in V_i$ and $v \in V_j$ , where $\{i,j\} \subseteq [k]$, $i< j$, we have the following preference lists:
    \par
    \smallskip
    \begin{tabularx}{10em}{@{} X  X @{}}
    $\basevertex{e}{}\colon$ & $\pref{\basevertex{\tilde{e}}{},\basevertex{u}{3},\basevertex{v}{3}, q_1^{ij}}$ \\
    $\basevertex{\tilde{e}}{}\colon$ & $\pref{\basevertex{e}{}, q_2^{ij}}$\\
     \end{tabularx}
     \par
     \medskip
    For each $i \in [k]$ and $\{i,j\}\subseteq [k]$, $i<j$, we have the following preference lists for the remaining vertices:
         \par
    \smallskip
    \begin{tabularx}{10em}{@{} X  X@{}}
    $p_1^i \colon$ & $\pref{\ordering{N(p_1^i)}}$
    \\
    $p_2^i \colon$ & $\pref{\ordering{N(p_2^i)}}$ 
    \\ 
    $q_1^{ij}\colon$ & $\pref{\ordering{N(q_1^{ij})}}$ 
    \\
    $q_2^{ij}\colon$ & $\pref{\ordering{ N(q_2^{ij}) }}$ 
    \\
    \end{tabularx}
   \caption{Preference lists in the constructed instance of \textsf{\textup{W[1]-hardness}} of \ams when parameterized by $k+q+t$. Here, for a set $S$, the symbol $\ordering{S}$ denotes that the vertices in this set are listed in some arbitrarily strict order and the notation $\pref{\cdot, \cdot}$ denotes the order of preference over neighbors. }
   \label{pref_list}
\end{table*}
\begin{figure*}
    \centering
\includegraphics[width=10cm,height=5cm,keepaspectratio]{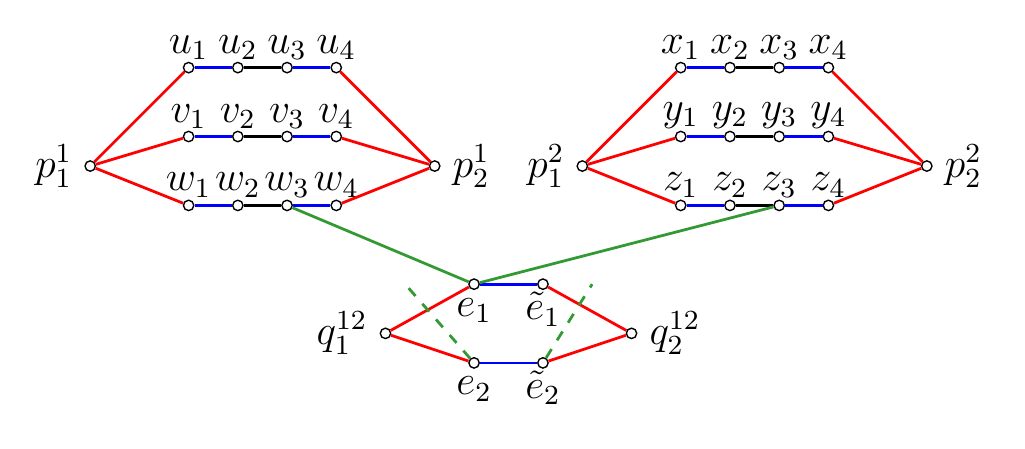}
\caption{An illustration of the construction of graph $G'$ in \W[1]-hardness of \ams. Here, {\color{blue} blue} colored edges belongs to the stable matching $\mu$. Note that $V_1=\{u,v,w\}$ and $V_2=\{x,y,z\}$, and $e_1$ and $e_2$ are edges in $E_{12}$.
}\label{fig:hardness}
\end{figure*}
\par
\noindent{\bf Correctness.} Next, we show the equivalence between the instance $\Co{I}$ of \mcqsmall and $\Co{J}$ of \ams. Formally, we prove the following:
\begin{sloppypar}
\begin{lemma}\label{lem:nasm-hardness-correcness}
$\Co{I}=(G,(V_1,\ldots,V_k))$ is a \yes-instance of \mcqsmall if and only if $\Co{J}=(G',\Co{L},\mu,k',q,t)$ is a \yes-instance of \ams.
\end{lemma}
\end{sloppypar}

\begin{proof}
In the forward direction, let $X$ be a solution of \mcqsmall for $\Co{I}$, i.e., for each $i\in [k]$, $|X\cap V_i|=1$, and $G[X]$ is a clique.
We construct a solution $\eta$ to $\Co{J}$ as follows. Initially, we set $\eta = \mu$. For each $i \in [k]$, if $u\in X\cap V_i$, we delete edges $\basevertex{u}{1}\basevertex{u}{2}$ and $\basevertex{u}{3}\basevertex{u}{4}$ from $\eta$, and add edges $\basevertex{u}{1}p_1^i,\basevertex{u}{2}\basevertex{u}{3}$, and $\basevertex{u}{4}p_2^i$ to $\eta$. Also, for each $\{i,j\} \sse [k]$, $i<j$, if $e\in E(G[X]) \cap E_{ij}$, then we remove the edge $\basevertex{e}{}\basevertex{\tilde{e}}{}$ from $\eta$ and add edges $\basevertex{e}{}q_1^{ij}$ and $\basevertex{\tilde{e}}{}q_2^{ij}$ to $\eta$. 
\begin{clm}\label{clm:nasm-eta-matching}
$\eta$ is a matching
\end{clm}
\begin{proof}
For each $i\in [k]$ and $u \in X\cap V_i$, 
edges $\basevertex{u}{1}p_1^i$, $\basevertex{u}{2}\basevertex{u}{3}$, and $\basevertex{u}{4}p_2^i$ are in $\eta$, and no other edge incident to $\basevertex{u}{1},\basevertex{u}{2},\basevertex{u}{3}$ or $\basevertex{u}{4}$ is in $\eta$. Since for each $i \in [k]$, $|X \cap V_i| = 1$, there is only one edge incident to each $p_1^i$ and $p_2^i$. Similarly, for each $e \in E(G[X])$, there is only one matching edge incident to $\basevertex{e}{}$ and $\basevertex{\tilde{e}}{}$, namely $\basevertex{e}{}q_1^{ij}$ and $\basevertex{\tilde{e}}{}q_2^{ij}$. Since, the remaning edges of $\eta$ are the same as in $\mu$,
 this implies $\eta$ is a matching.
\end{proof}

\begin{clm}\label{clm:nasm-eta-size}
$|\eta| = |\mu|+t$ and $|\mu \triangle \eta| = q$
\end{clm}
\begin{proof}
Note that for each $u\in X$, we delete two edges from $\eta$ (which also belongs to $\mu$), and add three edges to $\eta$. Similarly, for an edge $e\in E(G[X])$, we delete one edge from $\eta$ which is also in $\mu$, and add two edges to $\eta$. Hence, $|\eta|=|\mu|+k+\nicefrac{k(k-1)}{2}=|\mu|+t$, and $|\mu \triangle \eta|= 5k+\nicefrac{3k(k-1)}{2}=q$. 
\end{proof}
Next, we prove that $\eta$ has $k'=k+\nicefrac{k(k-1)}{2}$ blocking edges. 
 Due to Proposition ~\ref{cl:bp vertices}, to count the blocking edges with respect to $\eta$, we only investigate the vertices of $V(\mu \triangle \eta)$.  Note that 
 \begin{equation*}
 \begin{aligned}
 V(\mu \triangle \eta) = {} & \{\basevertex{u}{j},p_\ell^i \in V(G')\colon u\in X\cap V_i, j\in[4], \ell\in[2],i\in[k]\} \\
  & \bigcup \quad \{\basevertex{e}{}, \basevertex{\tilde{e}}{},q_\ell^{ij} \in V(G')\colon e\in E(G[X]), \{i,j\}\subseteq [k], i<j,\ell \in[2]\}  \end{aligned}
\end{equation*}
 \begin{clm}\label{clm:nasm-no bp on u3,u4}
 Let $u\in X$. There is no blocking edge incident to $\basevertex{u}{3}$ or $\basevertex{u}{4}$ with respect to $\eta$.
 \end{clm}
 \begin{proof}
 Since $\basevertex{u}{2}\basevertex{u}{3} \in \eta$, and $\basevertex{u}{3}$ prefers $\basevertex{u}{2}$ over any other vertex, there is no blocking edge incident to $\basevertex{u}{3}$. 
 Let $u \in V_i$, for some $i\in [k]$.
 Recall that the preference list of $\basevertex{u}{4}$ is $\pref{\basevertex{u}{3},p_2^i}$. Since there is no blocking edge incident to $u_3$ and $\basevertex{u}{4}p_2^i \in \eta$, it follows that there is no blocking edge incident to $\basevertex{u}{4}$.
 \end{proof}  
 \begin{clm}\label{clm:nasm-u1u2 bp}
 Let $u\in X$. Then, $\basevertex{u}{1}\basevertex{u}{2}$ is a blocking edge with respect to $\eta$. Moreover, there is no other blocking edge incident to $\basevertex{u}{1}$ or $\basevertex{u}{2}$.
 \end{clm}
 \begin{proof}
 Since $\basevertex{u}{1}$ and $\basevertex{u}{2}$ prefer each other over any other vertex and $\basevertex{u}{1}\basevertex{u}{2}\notin \eta$, it is a blocking edge with respect to $\eta$.
 Let $u \in V_i$, for some $i\in [k]$. Since the preference list of $\basevertex{u}{1}$ is $\pref{\basevertex{u}{2},p_1^i}$ and $\basevertex{u}{1}p_1^i \in \eta$, there is no other blocking edge incident to $\basevertex{u}{1}$. Similarly, since the preference list of $\basevertex{u}{2}$ is $\pref{\basevertex{u}{1},\basevertex{u}{3}}$ and $\basevertex{u}{2}\basevertex{u}{3} \in \eta$, it follows that there is no other blocking edge incident to  $\basevertex{u}{2}$. 
 \end{proof}
Using Claims~\ref{clm:nasm-no bp on u3,u4} and~\ref{clm:nasm-u1u2 bp}, for each $i\in [k]$ and $u\in X\cap V_i$, we introduce exactly one blocking edge with respect to $\eta$ by deleting $\basevertex{u}{1}\basevertex{u}{2}$ and $\basevertex{u}{3}\basevertex{u}{4}$ from $\eta$, and adding edges $\basevertex{u}{1}p_1^i,\basevertex{u}{2}\basevertex{u}{3}$, and $\basevertex{u}{4}p_2^i$ to it. Since $|X| = k$, in total we introduce $k$ blocking edges with respect to $\eta$ due to the said alternation.

 \begin{clm}
 For each $i\in [k]$, there is no blocking edge incident to $p_1^i$ or $p_2^i$ with respect to $\eta$. 
 \end{clm}
 \begin{proof}
 Let $u\in X\cap V_i$. Then, by the construction of $\eta$, $\basevertex{u}{4}p_2^i \in  \eta$. Let $v \in V_i \setminus \{u\}$. Since $|X\cap V_i| = 1$, $v \notin X$. Hence, $\basevertex{v}{4}\basevertex{v}{3} \in \eta$. Since $\basevertex{v}{4}$ prefers $\basevertex{v}{3}$ over $p_2^i$, $\basevertex{v}{4}p_2^i$ is not a blocking edge. Hence, there is no blocking edge incident to $p_2^i$ as $N(p_2^i)=\{\basevertex{w}{4}\colon w\in V_i\}$. Similarly, there is no blocking edge incident to $p_1^i$. 
 \end{proof}
 \begin{clm}
Let $e\in E(G[X])$. Then, $\basevertex{e}{}\basevertex{\tilde{e}}{}$ is a blocking edge with respect to $\eta$. Moreover, there is no other blocking edge incident to $\basevertex{e}{}$ or $\basevertex{\tilde{e}}{}$.
\end{clm}
\begin{proof}
 Since $\basevertex{e}{}\basevertex{\tilde{e}}{}\notin \eta$, and $\basevertex{e}{}$ and $\basevertex{\tilde{e}}{}$ prefer each other over any other vertex, $\basevertex{e}{}\basevertex{\tilde{e}}{}$ is a blocking edge with respect to $\eta$. Let $e = uv$ where $u\in V_i$, and $v\in V_j$. Recall that the preference list of $\basevertex{e}{}$ is $\pref{\basevertex{\tilde{e}}{},\basevertex{u}{3},\basevertex{v}{3},q_1^{ij}}$. Since $\basevertex{u}{3}$ does not prefer $\basevertex{e}{}$ over $\basevertex{u}{2}(=\eta(\basevertex{u}{3}))$, $\basevertex{u}{3}\basevertex{e}{}$ is not a blocking edge with respect to $\eta$. Similarly, $\basevertex{v}{3}\basevertex{e}{}$ is not a blocking edge with respect to $\eta$. 
 Since $\basevertex{e}{}q_1^{ij} \in \eta$, $\basevertex{e}{}\basevertex{\tilde{e}}{}$ is the only blocking edge incident to $\basevertex{e}{}$ for $\eta$. Since $N(\basevertex{\tilde{e}}{})=\{\basevertex{e}{},q_2^{ij}\}$, and $\basevertex{\tilde{e}}{}q_2^{ij}\in \eta$, there is no other blocking edge incident to $\basevertex{\tilde{e}}{}$ with respect to $\eta$. 
 \end{proof}
 \begin{clm}\label{clm:nasm-no bp on q1,q2}
 For each $\{i,j\}\subseteq [k], i<j$, there is no blocking edge incident to $q_1^{ij}$ or $q_2^{ij}$ with respect to $\eta$. 
 \end{clm}
 \begin{proof} 
 Let $e_1 \in E(G[X])\cap E_{ij}$. Hence, by the construction of $\eta$,
$\basevertex{{e_1}}{}q_1^{ij}, \basevertex{\tilde{e_1}}{}q_2^{ij} \in \eta$. Let $e_2 \in E_{ij} \setminus \{e_1\}$.
Since $|E(G[X])\cap E_{ij}| =1$, $e_2$ does not belong to $E(G[X])$.
Therefore, by the construction of $\eta$, $\basevertex{{e_2}}{}\basevertex{\tilde{e_2}}{}\in \eta$. 
 Since $\basevertex{{e_2}}{}$ and $\basevertex{\tilde{e_2}}{}$ prefer each other over any other vertex, there is no blocking edge incident to $\basevertex{{e_2}}{}$ or $\basevertex{\tilde{e_2}}{}$. Hence, $\basevertex{{e_2}}{}q_1^{ij}$ and $\basevertex{\tilde{e_1}}{}q_2^{ij}$ are not blocking edges. Since $N(q_1^{ij})=\{\basevertex{e}{} \in V(G')\colon e\in E_{ij}\}$, and $N(q_2^{ij})=\{\basevertex{\tilde{e}}{} \in V(G')\colon e\in E_{ij}\}$, there is no blocking edge incident to $q_1^{ij}$ or $q_2^{ij}$.
 \end{proof}
 Hence, for each $e\in E(G[X])$, we introduce one blocking edge $\basevertex{e}{}\basevertex{\tilde{e}}{}$ with respect to $\eta$. That is, we introduce $\nicefrac{k(k-1)}{2}$ blocking edges. Using Claims \ref{clm:nasm-no bp on u3,u4} to \ref{clm:nasm-no bp on q1,q2}, there are $k+\nicefrac{k(k-1)}{2}$ blocking edges for $\eta$. This completes the proof in the forward direction.
\par
In the reverse direction, let $\eta$ be a matching of size at least $|\mu|+t$ such that $|\mu \triangle \eta| \leq 5k+\nicefrac{3k(k-1)}{2}$, and $\eta$ has at most $k'$ blocking edges. Recall  that  $|V(G')|=4\lvert V(G)\rvert+2\lvert E(G)\rvert+2k+k(k-1)$, $\mu=2\lvert V(G)\rvert + \lvert E(G)\rvert$, and $t = k+\nicefrac{k(k-1)}{2}$. Hence, $\eta$ is a perfect matching in $G'$. 

Note that, similar to Theorem~\ref{thm:asm}, in our instance, the static edges in $G'$ are of the following type: For any $ u \in V (G)$, edge $\basevertex{u}{1}\basevertex{u}{2}$ in $G'$ is a static edge and is called the $u$-type static edge; for any $e \in E(G)$, edge $\basevertex{e}{}\basevertex{\tilde{e}}{}$ in $G'$ is a static edge and is called the $e$-type static edge.

  Let ${\mathtt B}_{\eta}$ be the set of blocking edges with respect to $\eta$.  Let us  note the following properties of the set ${\mathtt B}_{\eta}$. Specifically we show that an edge in ${\mathtt B}_\eta$ is either a $u$-type static edge or an $e$-type static edge. In fact, for each $i \in [k]$, there is a unique $u$-type static edge which is a blocking edge, and for each $\{i,j\} \sse [k]$, there is a unique $e$-type static edge in ${\mathtt B}_\eta$.
\begin{clm}[$u$-type static edge]\label{clm:nasm-claim:solution intersect each part}
For each $i\in [k]$, there exists a vertex $u \in V_i$  such that the edge $\basevertex{u}{1}\basevertex{u}{2}\in {\mathtt B}_{\eta}$.
\end{clm}   
\begin{proof}
Since $\eta$ is a perfect matching, $p_1^i$ is saturated by $\eta$, for each $i\in [k]$. Since $N(p_1^i)=\{\basevertex{u}{1} \colon u \in V_i\}$, we have that $p_1^i\basevertex{u}{1} \in \eta$, for some $u \in V_i$. Since $\basevertex{u}{1}$ and $\basevertex{u}{2}$ prefer each other over any other vertex, it follows that $\basevertex{u}{1}\basevertex{u}{2} \in {\mathtt B}_{\eta}$. 
\end{proof}

\begin{clm}[$e$-type static edge]\label{clm:nasm-claim:solution intersect clique edge}
For each $\{i,j\}\subseteq [k]$, there exists an edge $e \in E_{ij}$ such that the edge $\basevertex{e}{}\basevertex{\tilde{e}}{}\in {\mathtt B}_{\eta}$.
\end{clm}  
\begin{proof}
Since $\eta$ is a perfect matching, $q_1^{ij}$ is saturated by $\eta$, for each $\{i,j\}\subseteq [k]$, where $i<j$. Since $N(q_1^{ij}) = \{ e \in V(G') \colon e \in E_{ij}\}$, $\basevertex{e}{}q_1^{ij} \in \eta$, for some $\basevertex{e}{}\in V(G')$. 
Since $\basevertex{e}{}$ and $\basevertex{\tilde{e}}{}$ prefer each other over any other vertex, it follows that $\basevertex{e}{}\basevertex{\tilde{e}}{} \in {\mathtt B}_{\eta}$. 
\end{proof} 
Using Claims ~\ref{clm:nasm-claim:solution intersect each part} and ~\ref{clm:nasm-claim:solution intersect clique edge}, and the fact that $|{\mathtt B}_{\eta}|=k+\nicefrac{k(k-1)}{2}$, we have following two properties of ${\mathtt B}_{\eta}$.
\begin{corollary}\label{clm:nasm-obs:size B}
 For each $i\in [k]$, there exists a unique vertex $u \in V_i$  such that the edge $\basevertex{u}{1}\basevertex{u}{2}\in {\mathtt B}_{\eta}$; and for each $\{i,j\}\subseteq [k]$ where $i<j$, there exists a unique edge $e \in E_{ij}$ such that the edge $\basevertex{e}{}\basevertex{\tilde{e}}{}\in {\mathtt B}_{\eta}$.
\end{corollary}

\begin{corollary}\label{clm:nasm-obs:type of bp}
Any edge in the set ${\mathtt B}_{\eta}$ 
is either a $u$-type static edge or an $e$-type static edge.

\end{corollary}

Next, we note a property that forces an edges in the matching $\eta$. 
\begin{clm}\label{clm:nasm-q1 pair}
 For any $\{i,j\}\subseteq [k]$, consider some $e \in E_{ij}$ such that $\basevertex{e}{}\basevertex{\tilde{e}}{} \in {\mathtt B}_{\eta}$. Then $q_1^{ij}\basevertex{e}{}\in \eta$.
\end{clm}
\begin{proof}
Suppose $q_1^{ij}\basevertex{e}{}\notin \eta$, then since $\eta$ is a perfect matching, there exists a vertex $e' \in V(G')$ such that $q_1^{ij}\basevertex{e'}{} \in \eta$. Since $\basevertex{e'}{}$ and $\basevertex{\tilde{e'}}{}$ prefer each other over any other vertex, $\basevertex{e'}{}\basevertex{\tilde{e'}}{} \in {\mathtt B}_{\eta}$. Recall that $N(q_1^{ij})  = \{ e' \in V(G') \colon e' \in E_{ij}\}$. Therefore, $\basevertex{e'}{} \in E_{ij}$,  a contradiction to 
the uniqueness criteria in Corollary \ref{clm:nasm-obs:size B}. Therefore, $q_1^{ij}\basevertex{e}{}\in \eta$.
\end{proof}
\begin{clm}\label{clm:nasm-consistency}
 Let $e=uv$, $u\in V_i$ and $v\in V_j$ where $\{i,j\}\subseteq [k]$ and $i<j$. If $\basevertex{e}{}\basevertex{\tilde{e}}{} \in {\mathtt B}_{\eta}$, then  $\{\basevertex{u}{1}\basevertex{u}{2},\basevertex{v}{1}\basevertex{v}{2}\}\subseteq {\mathtt B}_{\eta}$, 
\end{clm}
\begin{proof}
We first show that $\basevertex{u}{1}\basevertex{u}{2} \in {\mathtt B}_{\eta}$. Recall that the preference list of $\basevertex{u}{3}$ is $\pref{\basevertex{u}{2}, \ordering{\X{E}_u}, \basevertex{u}{4}}$. If $\eta(\basevertex{u}{3}) = \basevertex{u}{2}$, then since $\basevertex{u}{1}$ and $\basevertex{u}{2}$ prefer each other over any other vertex, $\basevertex{u}{1}\basevertex{u}{2} \in {\mathtt B}_{\eta}$.

Suppose that $\eta(\basevertex{u}{3}) = \basevertex{e'}{}$ where $\basevertex{e'}{} \in \X{E}_u$, then since $\basevertex{e'}{}$ and $\basevertex{\tilde{e'}}{}$ prefer each other over any other vertex, $\basevertex{e'}{}\basevertex{\tilde{e'}}{} \in {\mathtt B}_{\eta}$. Since $\basevertex{u}{3} \basevertex{e'}{} \in \eta$,  we get a contradiction to Claim~\ref{clm:nasm-q1 pair}. 
Therefore, $\eta(\basevertex{u}{3}) \notin \X{E}_u$. Since $\eta$ is a perfect matching, $\eta(\basevertex{u}{3}) = \basevertex{u}{4}$.
Note that $\basevertex{u}{3}$ prefers the vertex $\basevertex{e}{}$ over $\basevertex{u}{4}$. Since, $\basevertex{e}{}\basevertex{\tilde{e}}{} \in {\mathtt B}_{\eta}$, by Claim ~\ref{clm:nasm-q1 pair} we have that $q_1^{ij}\basevertex{e}{}\in \eta$.
Note that $\basevertex{e}{}$ also prefers $\basevertex{u}{3}$ over $q_1^{ij}$. Therefore, $\basevertex{u}{3}\basevertex{e}{} \in {\mathtt B}_{\eta}$. This contradicts  Corollary~\ref{clm:nasm-obs:type of bp}. Similarly, we can show that $\basevertex{v}{1}\basevertex{v}{2} \in {\mathtt B}_{\eta}$.
\end{proof}
Next, we construct two sets $S$ and $E_S$ as follows. Let $S=\{u\in V(G)\colon \basevertex{u}{1}\basevertex{u}{2} \in {\mathtt B}_{\eta}, i\in[k]\}$, and $E_S=\{e\in E(G)\colon \basevertex{e}{}\basevertex{\tilde{e}}{}\in {\mathtt B}_{\eta}, \{i,j\}\subseteq [k]\}$. We claim that $G_S=(S,E_S)$ is a clique, and $|S\cap V_i|=1$, where $i\in [k]$. Let $e=uv$, where $u\in V_i$, and $v\in V_j$. Using Claim \ref{clm:nasm-consistency}, for each $\basevertex{e}{}\basevertex{\tilde{e}}{}\in {\mathtt B}_{\eta}$, $\{\basevertex{u}{1}\basevertex{u}{2},\basevertex{v}{1}\basevertex{v}{2}\} \subseteq {\mathtt B}_{\eta}$. Hence for each $uv\in E_S$, $\{u,v\}\subseteq S$. Using Corollary  
\ref{clm:nasm-obs:size B}, 
$|S\cap V_i|=1$, i.e., $|S|=k$
, and 
$|E_S|=\nicefrac{k(k-1)}{2}$. 
Hence, $G_S$ is a clique. This completes the proof.

\begin{supress}
\begin{lemma}\label{lem:restricted_domain2}
Even in the presence of master lists on both side,  {\sf ASM} and \ams are \WOH with respect to $k + q+t$. 
\end{lemma}
\begin{proof}
We show that the preference lists of the vertices in each side of the partition, in the construction in the proof of Theorem~\ref{thm:smd w-hard}, respects a master ordering of vertices. Recall the bipartite graph constructed in the  proof of Theorem~\ref{thm:smd w-hard}. Consider the partition that contains $p_1^i,\basevertex{v}{2},\basevertex{v}{4},\basevertex{e}{},q_2^{ij}$, where $\{i,j\}\subseteq [k]$. Let us call this partition as $X$. Let $P=\{p_1^i \colon i\in [k]\}$, $V_{even}=\{\basevertex{v}{2},\basevertex{v}{4} \colon i\in [k]\}$, $E_1=\{\basevertex{e}{}\colon \{i,j\}\in [k]\}$, and $Q_2=\{q_2^{ij}\colon \{i,j\}\subseteq [k]\}$.  
Note that no two vertices of the set $P$ are adjacent to a vertex in $G'$. Therefore, we can place vertices of $P$ in any arbitrary order.  We place the vertices of $U_{odd}$, in the sorted order of $i\in [k]$ as two vertices of $U_{odd}$ can be adjacent to $\basevertex{e}{}$, where $\{i,j\}\subseteq [k], i<j$. Since no two vertices of $E_2$ have common neighbor in $G'$, we can place vertices of $E_2$ also in any arbitrary order. Similarly, we can place vertices of $Q$ in any arbitrary order. Now, in the master list corresponding to vertices in the set $X$, we first place the vertices of $E_2$ followed by vertices of $U_{odd}$ in the above mentioned order. Similarly, we can construct master list of the other part in the bipartition.
\end{proof}
\end{supress}
\end{proof}

\section{\fpt Algorithm for \ams}\label{sec:fpt}
In this section, we give \FPT\ algorithm for \ams with respect to $q+d$ (Theorem~\ref{th:fptAlgo}). Recall that $d$ is the degree of the graph $G$, and $q$ is the symmetric difference between a solution matching and the given stable matching $\mu$. Suppose  $\eta$ is a {\em hypothetical} solution to $(G,\Co{L},\mu,k,q,t)$. Let matchings $\mu=\mu_1 \uplus \mu_2$ and $\eta =\mu_2 \uplus \eta_2$. Observe that we can obtain  
$\eta$ from $\mu$, by deleting $\mu_2$, and adding the edges in $\eta_2$. Equivalently, we can find $\eta$, if we know $\mu \triangle \eta$, as $\mu \triangle \eta = \mu_{2} \uplus \eta_{2}$. Thus, our goal is reduced to find $\mu \triangle \eta$.  Now, we begin with the description of our algorithm, which has three phases: Vertex Separation, Edge Separation, and Size-Fitting. An example describes the algorithm in Figure~\ref{fig:algo}. We begin with the description of a randomized algorithm which will be derandomized later using $n$-$p$-$q$-{\em lopsided universal} family~\cite{FominLPS16}. Given an instance $(G,\Co{L},\mu,k,q,t)$ of \ams, we proceed as follows.

\medskip
\noindent
{\bf \phase I: Vertex Separation}\label{phase1}

\begin{tcolorbox}[colback=red!5!white,colframe=red!75!black]
Let $f$ be a function that colors each vertex of the graph  $G$ independently with color $1$ or $2$ with probability $1/2$ each. 
\end{tcolorbox} 

Then, the following properties hold for $G$ that is colored using the function $f$:
    
\begin{itemize}
\item  Every vertex in $V(\mu \triangle \eta)$ is colored $1$ 
with probability at least $\frac{1}{2^{2q}}$. 
\item Let $B$ be a set of the neighbors of the vertices in $V(\mu \triangle \eta)$ outside the set $V(\mu \triangle \eta)$, that is, $B=N_{G}(V(\mu \triangle \eta))$, and $D$ be the set of matching partners of the vertices in $B$, in the matching $\mu$, if they exist. Every vertex in $B \cup D$ is colored $2$ 
with probability at least $\frac{1}{2^{4qd}}$. To see this note that  $|\mu\triangle \eta| \leq q$ and the  maximum degree of a vertex in the graph $G$ is $d$, and so $|B \cup D|\leq 2|B| = 2|N_{G}(V(\mu \triangle \eta))|  \leq 4qd$.
\end{itemize}
\smallskip
For $i\in [2]$, let $V_i$ denote the set of vertices of the graph $G$, that are colored $i$ using the function $f$. Summarizing the above mentioned properties we get the following. 
\begin{lemma}\label{label:vertexsep}
Let $V_{1},V_2$, $B$ and $D$ be as defined above. Then, with probability at least $\frac{1}{2^{2q+4qd}}$, 
$V(\mu \triangle \eta) \subseteq V_1$ and $B\cup D \sse V_2$. 
\end{lemma}

Due to Lemma~\ref{label:vertexsep}, we have the following: 
\begin{corollary}\label{obs:component of colored graph}
 Every component in $G[V(\mu \triangle \eta)]$ is also a component in $G[V_{1}]$ with probability at least $\frac{1}{2^{2q+4qd}}$.
\end{corollary}

The proof of Corollary~\ref{obs:component of colored graph} follows from the fact that $V(\mu \triangle \eta) \subseteq V_1$ and $B = N_G(V(\mu \triangle \eta))$ is a subset of $V_2$. Thus, due to Corollary~\ref{obs:component of colored graph}, 
if there exists a component in $C$ containing a vertex $u \in V(G)$ such that $\mu(u) \notin C$, then $C$ is not a component in $G[V(\mu \triangle \eta)]$. Thus, we get the following reduction rule. 
 
\begin{reduction}
If there exists a component in $C$ containing a vertex $u \in V(G)$ such that $\mu(u) \notin C$, then 
delete the component $C$ from $G[V_1]$. 
\end{reduction}

In light of Corollary~\ref{obs:component of colored graph}, to find $\mu \triangle \eta$, 
 in \phase~II, we color the edges of $G[V_{1}]$ in order to identify the components of the graph that {\em only contains edges of  $\mu\triangle\eta$.} 

\begin{figure}
\begin{center}
\includegraphics[scale=1]{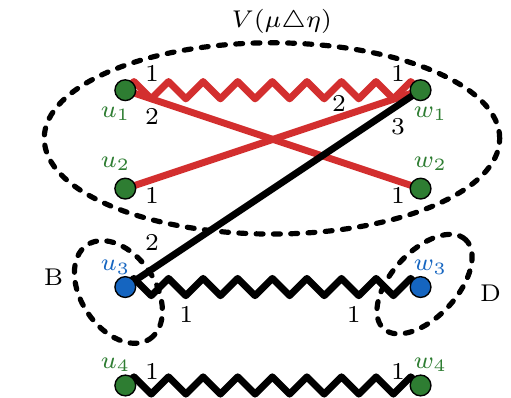}
\end{center}
\caption{The zigzag edges represent the edges of the stable matching $\mu$. The matching $\eta=\{u_1w_2, u_2w_1, u_3w_3, u_4w_4\}$, and sets $B$ and $D$ are as defined in the \phase I of the algorithm. Vertex colors $1$ and $2$ in \phase I are represented by  \textcolor{mygreen}{\bf{green}} and \textcolor{rvwvcq}{{\bf blue}}, respectively. Hence, $G_{1}= G[\{u_1,u_2,u_4,w_1,w_2,w_4\}]$. The \textcolor{dtsfsf}{{\bf red}} edges represent the edges in $\mu\triangle\eta$ in \phase II.} 
\label{fig:algo}
\end{figure}

\medskip
\noindent
{\bf \phase II: Edge Separation\label{phase2}}

\begin{tcolorbox}[colback=red!5!white,colframe=red!75!black]
Let $g$ be a function that colors each edge of the subgraph $G[V_{1}]$ independently with colors $1$ or $2$ \wp $1/2$ each. 
\end{tcolorbox}

 Let $G_{1}=G[V_{1}]$ and let $G'=G_{1}[V(\mu \triangle \eta)]$. Then, the following properties hold for the graph $G_1$ that is colored using the function $g$: 
   
\begin{itemize}
\item Every edge in $\mu \triangle \eta$ is colored $1$ 
\wp at least $\frac{1}{2^q}$. 
\item Every edge in $E(G')\setminus\!(\mu\triangle\eta)$ is colored $2$  
\wp at least $\frac{1}{2^{2qd}}$, because $|V(\mu \triangle \eta)|\leq 2q$ and $d$ is the maximum degree of a vertex in the graph $G$, so $|E(G')|\leq 2qd$. 
\end{itemize}

 For $i\in [2]$, let $E_i$ denote the set of edges of the graph $G_1$ that are colored $i$ using the function $g$.   
Then, due to the above mentioned coloring properties of the graph $G_1$, we have the following result:

\begin{lemma}\label{label:edgesep}
Let $G'$, $E_1$, and $E_2$ be as defined above. Then, with probability at least $\frac{1}{2^{q+2qd}}$, 
$\mu \triangle \eta \subseteq E_1$ and $E(G')\setminus\!(\mu\triangle\eta) \sse E_2$. 
\end{lemma}

Note that the edges in $\mu\triangle\eta$ form $\mu$-alternating paths/cycles. Therefore, if there exists a component $C$ in $G_1$ such that the set of colored $1$ edges in $C$ do not form a $\mu$-alternating path or a cycle, then  we could delete this component from $G_1$.

\begin{reduction}
\label{rule-clean-ap-ac}
If there exists a component in $C$ containing a vertex $u \in V(G)$ such that $\mu(u) \notin C$, then 
delete the component $C$ from $G[V_1]$. 
\end{reduction}

 Let $G^\star=(V_1,E_1)$ be a graph on which Reduction Rule~\ref{rule-clean-ap-ac} is not applicable. Then, we get the following. 
\begin{observation}\label{obs:G*}
Every component in $G^\star$ is a $\mu$-alternating path/cycle
\end{observation}

The next lemma ensures that we have highlighted our solution with good probability. The proof of it follows from Lemmas~\ref{label:vertexsep} and \ref{label:edgesep}. 

\begin{lemma}\label{label:2coloredgesinsolA}
Let $(G,\Co{L},\mu,k,q,t)$ be a \yes-instance of \ams. Then \wp at least $\frac{1}{2^{3q+6qd}}$, there exists a solution $\eta$ such that (a) it contains every edge of $\mu$ whose both the endpoints are colored $2$ by $f$, and  (b) there exists a family of components $\mathscr{C}$ of $G^\star$ such that $\eta$ contains all the edges in $\mathscr{C}$ that do not belong to $\mu$ but are colored $1$ by $g$.
\end{lemma}


In light of Lemma~\ref{label:2coloredgesinsolA}, our goal is reduced to find a family of components $\mathscr{C}$ of $G_1$ that contains the edges of $\mu \triangle \eta$. Due to Observation~\ref{obs:G*}, to obtain a matching of size $|\mu|+t$, we can choose $t$ components of $G^\star$ which are $\mu$-augmenting paths (an alternating path, a path that alternates between matching and a non-matching edge, where the first and the last edge  are non-matching edge).  However, choosing $t$ components arbitrarily might lead to a large number of blocking edges in the matching $\eta$. Thus, to choose the components of $G^\star$ appropriately, 
 we move to \phase III. 
\par
\medskip
\noindent
{\bf \phase III: \fit with respect to $g$.}\label{phase3}
In this phase, we proceed with the function $g$ and the graph  
$G^\star$ obtained after \phase II (that is the one where every component satisfies the property that edges which are colored $1$ form a $\mu$-alternating path/cycle). 
Next, we will reduce the instance to an instance of \wtknapsack (\wtkp), and after that 
use an algorithm for \wtkp, described in Proposition~\ref{lem:knapsack}, as a subroutine. 

 \defproblem{\wtknapsack(\wtkp)}
{A set of tuples, $\Co{X}=\{ (a_i,b_i,p_i) \in \mathbb{N}^{3}: i \in [n]\}$, and non-negative integers $c_1,c_2$ and $p$}{Does there exist a set $Z \subseteq [n]$ such that $\sum_{i \in Z} a_i \leq c_1$, $\sum_{i \in Z} b_i \leq c_2$, and $\sum_{i \in Z} p_i \geq p$?}  
 
\begin{sloppypar}
\begin{prop}{\rm\cite{DBLP:books/daglib/0010031}}\label{lem:knapsack} There exists an algorithm \Co{A} that given an instance $(\Co{X},c_1,c_2,p)$ of \wtkp, in time $\OO(nc_1c_2)$, outputs a solution if it is a \yes-instance of \wtkp; otherwise \Co{A} outputs ``{\rm{no}}''.
\end{prop}
\end{sloppypar}

Next, we construct an instance of \wtkp as follows. 
 Let $C_1,\ldots, C_\ell$ be the components of the graph $G^\star$. For each $i\in [\ell]$,  we compute the number of blocking edges, $k_i$, incident on the vertices in $C_i$ by constructing a matching $\eta_i$ as follows. We first add all the edges inside the component $C_i$ which are not in $\mu$, to $\eta_i$. Further, we add all the edges in $\mu$ which are not in $C_i$  and whose at least one of the endpoint is a neighbor of a vertex in $C_i$. Clearly, $\eta_i$ is a matching in the graph $G$. We set $k_i$ as the number of blocking edges with respect to $\eta_i$.   Let $q_i$ denote the number of edges in $C_i$, where $i\in [\ell]$. 
  Let $\mu_i \subseteq \mu$ be the set edges in $C_i$, where $i\in [\ell]$. 
  For each $i\in [\ell]$, let $t_i=q_{i}-2|\mu_{i}|$. Intuitively, $t_i$ denote the increase in the size of the matching, if we include the $\mu$-alternating path/cycle in $C_i$ to the solution matching $\eta$
 
 Let $\Co{X}=\{ (k_i,q_i,t_i): i \in [\ell]\}$.  This gives us an instance $(\Co{X},k,q,t)$ of \wtkp. 
 We invoke algorithm $\Co{A}$ given in Proposition~\ref{lem:knapsack} on the instance $(\Co{X},k,q,t)$ of \wtkp. If $\Co{A}$ returns a set $Z$, then we return  ``yes''. 
 Otherwise, we report {\em failure} of the algorithm.
 It is relatively  straightforward to create the solution $\eta$ when the answer is ``yes''. 

\begin{lemma}
\label{lem:phase3}Let $(G,\Co{L},\mu,k,q,t)$ be a \yes-instance of \ams. Then, \wp at least $\frac{1}{2^{3q+6qd}}$, we return ``yes''. 
\end{lemma}
\begin{proof}
Let $\eta$ be a solution claimed in the statement of Lemma~\ref{label:2coloredgesinsolA}. Let $\mathscr{C}$ be the family of components mentioned in the statement of Lemma~\ref{label:2coloredgesinsolA}.
Recall that $C_1,\ldots, C_\ell$ are the components of the graph $G^\star$. We next show that $S=\{i\in [\ell]  \colon C_i \in \mathscr{C}\}$ is a solution to  
$(\Co{X},k,q,t)$. 
 Due to property (b) of the solution $\eta$ and the construction of the instance $(\chi,k,q,t)$, $\sum_{C_i \in \mathscr{C}}q_i \leq q$ and $\sum_{C_i \in \mathscr{C}}t_i \geq t$.  We next show that $\sum_{C_i \in \mathscr{C}}k_i \leq k$. Consider a component $C_i$ in $\mathscr{C}$. We first recall that if $C_i$ is a component in $G[V(\mu \triangle \eta)]$, then $N(V(C_i))$ and matching partners of the vertices in $N(V(C_i))$, in the matching $\mu$ are colored $2$ by $f$ with probability at least $\frac{1}{2^{4qd}}$. Thus, $\eta_i \subseteq \eta$, by the construction of $\eta_i$.  
We next show that every blocking edge with respect to $\eta_i$, where $C_i$ is a component in $\mathscr{C}$, is also a blocking edge with respect to $\eta$. Let $uv$ be a blocking edge in $\eta_i$. Then, $v \succ_{u} \eta_i(u)$ and  $u \succ_{v} \eta_i(v)$. Since $\eta_i \subseteq \eta$, it follows that $v \succ_{u} \eta(u)$ and  $u \succ_{v} \eta(v)$. Hence, $uv$ is also a blocking edge with respect to $\eta$. Since $k_i$ is the the number of blocking edges with respect to $\eta_i$, we can infer that $\sum_{C_i \in \mathscr{C}}k_i \leq k$. Hence, $(\Co{X},k,q,t)$ is a \yes-instance of \wtkp. Therefore, due to Proposition~\ref{lem:knapsack}, we return ``yes".   
\end{proof}

\begin{sloppypar}
\begin{lemma}\label{lem:kp-to-ams-proof}
Suppose that $(\chi,k,q,t)$ is a \yes-instance of \wtkp. Then, $(G,\Co{L},\mu,k,q,t)$ is a \yes-instance of \ams. 
\end{lemma}
\end{sloppypar}

\begin{proof}
Suppose that the algorithm $\Co{A}$ in Proposition~\ref{lem:knapsack} returns the set $Z$. 
Given the set $Z$, we obtain the matching $\eta$ as follows. 
Let $Z(\mathscr{C})$ denote the family of components of $G^\star$ corresponding to the indices in $Z$. Formally, $Z(\mathscr{C})=\{C_i \colon i \in Z \text{ and } C_i \text{ is a component of } G^\star\}$. For each component $C \in Z(\mathscr{C})$, we add all the edges in $C$ that are not in $\mu$, to $\eta$. Additionally, we add all the edges in $\mu$ to $\eta$, whose both the endpoints are outside the components in $Z(\mathscr{C})$. We next prove that $\eta$ is a solution to $(G,\Co{L},\mu,k,q,t)$. 
\begin{clm}\label{clm:eta-matching}
$\eta$ is a matching.
\end{clm}

\begin{proof}
Towards the contradiction, suppose that $uv,uw \in \eta$, that is, there exists a pair of edges in $\eta$ that shares an endpoint.  Note that  $uv$ and $uw$ cannot be in two different components of $G^\star$ by the construction of the graph $G^\star$.   If $uv$ and $uw$ both are in the same component $C \in Z(\mathscr{C})$, then it contradicts Observation~\ref{obs:G*} as $C$ is also a component in $G^\star$.  Suppose that $uv \in \mu$ but not in any component in $G^\star$. We claim that there is no component in  $G^\star$ containing $uw$. Towards the contradiction, let $C$ be a component in $G^\star$ that contains $uw$. Clearly, $C$ is also a component in $G[V_1]$. This contradicts the fact that in \phase I, we have deleted the component $C$ as it contains a vertex $u \in V(G)$ such that $\mu(u) \notin C$. Since $uw \in \eta$ but $uw$ is not in any component in $G^\star$, it follows that $uw \in \mu$, by the construction of $\eta$. Since $uv,uw \in \mu$, it contradicts the fact that $\mu$ is a matching. 
\end{proof}

\begin{clm}\label{clm:eta-size-bound}
$|\eta| \geq |\mu|+t$ and $|\mu \triangle \eta| \leq q$.
\end{clm}

\begin{proof}
For each $C_i \in Z(\mathscr{C})$, let $\mu_i = \mu \cap E(C_i)$, that is, $\mu_i$ is the set of edges in $C_i$ that are in $\mu$. Let $\tilde{\mu}$ be the set of edges in $\mu$ that does not belong to any component in $Z(\mathscr{C})$. Thus, $\mu=\uplus_{C_i \in Z(\mathscr{C})} \mu_i \uplus \tilde{\mu}$.  We first show that if $uv \in \tilde{\mu}$, then both $u$ and $v$  do not belong to any component in $Z(\mathscr{C})$, because if $u$ or $v$ belong to a component $C$ in $Z(\mathscr{C})$, then as argued above it contradicts the fact that we have deleted $C$ in \phase I. Thus, by the construction of $\eta$, $\tilde{\mu} \subseteq \eta$. Furthermore, $\eta= \uplus_{C_i \in Z(\mathscr{C})} (E(C_i)\setminus \mu_i) \uplus \tilde{\mu}$. Since every $C_i \in Z(\mathscr{C})$ is a $\mu$-alternating path due to Observation~\ref{obs:G*}, we have that $|\mu \triangle \eta| = \sum_{C_i \in Z(\mathscr{C})} q_i \leq q$ as $Z$ is a solution to $(\chi,k,q,t)$. Furthermore, $|\eta|=|\tilde{\mu}|+\sum_{C_i \in Z(\mathscr{C})} (E(C_i)\setminus \mu_i) = |\tilde{\mu}|+\sum_{C_i \in Z(\mathscr{C})} (q_i-|\mu_i|) = |\tilde{\mu}|+\sum_{C_i \in Z(\mathscr{C})} (t_i+|\mu_i|)$. Since $\sum_{i\in Z}t_i \geq t$, we obtained that $|\eta|\geq |\mu|+t$.
\end{proof}

\begin{clm}\label{clm:eta-bp}
There are at most $k$ blocking edges with respect to $\eta$.
\end{clm}

\begin{proof}
For a component $C_i$ in $G^\star$, recall the definition of $\eta_i$ in \phase III. $\eta_i$ contains all the edges in $C_i$ which are not in $\mu$ and also the edges which are in $\mu$ but not in $C_i$ 
 and whose at least one of the endpoint is a neighbor of a vertex in $C_i$. We first prove that every blocking edge with respect to the matching $\eta$ is also a blocking edge with respect to matching $\eta_i$, for some $C_i \in Z(\mathscr{C})$. Let $uv$ be a blocking edge with respect to $\eta$. Due to Proposition~\ref{cl:bp vertices} and by the construction of $\eta$, either $u$ or $v$ belongs to a component in $Z(\mathscr{C})$. Without loss of generality, let $u$ belongs to a component $C_i \in Z(\mathscr{C})$. Thus, $\eta(u)=\eta_i(u)$, by the construction of $\eta$ and $\eta_i$. If $v$ is also in $C_i$, then $\eta(v)=\eta_i(v)$, and hence $uv$ is a blocking edge with respect to $\eta_i$. Suppose that $v \notin C_i$. Since $uv \in E(G)$, by the construction of the graph $G_1$, $v$ does not belong to any other component of $G_1$. Thus, by the construction of $\eta$ and $\eta_i$, $\eta(v)=\mu(v)$ and $\eta_i(v)=\mu(v)$. Therefore, $uv$ is also a  blocking edge with respect to $\eta_i$. Recall that $k_i$ is the number of blocking edges with respect to $\eta_i$. Therefore, the number of blocking edges with respect to $\eta$ is at most $\sum_{i \in Z}k_i \leq k$.
 \end{proof}
 
 Due to Claims~\ref{clm:eta-matching},~\ref{clm:eta-size-bound}, and~\ref{clm:eta-bp}, we can infer that $\eta$ is a solution to $(G,\Co{L},\mu,k,q,t)$. 
\end{proof}

 Due to Lemmas~\ref{lem:phase3} and~\ref{lem:kp-to-ams-proof}, we obtain a polynomial-time randomized algorithm for \ams\ which succeeds \wp $\frac{1}{2^{3q+6qd}}$. Therefore, by repeating the algorithm independently $2^{3q+6dq} (\log n)^{\OO(1)}$ times, where $n$ is the number of vertices in the graph, we obtain the following result:
 
 \begin{theorem}
 There exists a randomized algorithm that given an instance of \ams runs in $2^{3q+6dq}n^{\OO(1)}$ time, where $n$ is the number of vertices in the given graph, and either reports a failure or outputs ``yes''. Moreover, if the algorithm is given a \yes-instance of the problem, then it returns ``yes'' with a constant probability.
 \end{theorem}


\subsection{Deterministic \FPT\ algorithm}
To make our algorithm deterministic we first introduce the notion of an $n$-$p$-$q$-{\em lopsided universal} family. Given a universe $U$ and an integer $\ell$, we denote all the $\ell$-sized subsets of $U$ by ${U \choose \ell}$. We say that a  family 
$\mathcal{F}$ of sets over a universe $U$ with $\vert U\vert=n$, is an $n$-$p$-$q$-{\em lopsided universal} family if for every $A \in {U \choose p}$ and $B \in {U \setminus A \choose q}$, there is an $F \in \mathcal{F}$ such that $A \subseteq F$ and $B \cap F = \emptyset$.

\begin{lemma}[\cite{FominLPS16}]
\label{lem:lopsidedUniversal}
There is an algorithm that given $n,p,q\in {\mathbb N}$ constructs an  $n$-$p$-$q$-lopsided universal family $\mathcal{F}$ of cardinality ${p+q \choose p} \cdot 2^{o(p+q)}  \log n$ in time 
$\vert  \mathcal{F} \vert  n$. 
\end{lemma}

\noindent{\bf{Algorithm:}}
Let $n$ and $m$ to denote  the number of vertices and edges in the given graph, respectively. To replace the function $f$ in our algorithm, we use an $n$-$2q$-$4qd$-lopsided universal family $\mathcal{F}_1$ of cardinality ${2q+4qd \choose q} \cdot 2^{o(dq)}  \log n$, where $\mathcal{F}_1$ is a family over the vertex set of $G$. To replace the function $g$, we use  $m$-$q$-$2qd$-lopsided universal family $\mathcal{F}_2$ of cardinality ${q+2qd \choose q} \cdot 2^{o(dq)}  \log m$, where $\mathcal{F}_2$ is a family over the edge set of $G$. For every set $F\in \mathcal{F}_1$ we create a function $f_F$ that colors every vertex of $F$ as $1$, and colors all the other vertices as $2$. Similarly, for every set $F\in \mathcal{F}_1$, we create a  function $g_F$ that colors every edge of $F$ as $1$, and  colors all the other edges as $2$. Now, for every pair of functions $(f_F,g_{F'})$, where $F\in \mathcal{F}_1$ and $F'\in \mathcal{F'}$, we run our algorithm described above. If for any pair of function $(f_F,g_{F'})$, where $F\in \mathcal{F}_1$ and $F'\in \mathcal{F'}$, the algorithm returns ``yes'', then we return ``yes'', otherwise ``no''. 

\par
\smallskip

\noindent{{\bf Correctness and Running Time:}}
Suppose that $(G,\Co{L},\mu,k,q,t)$ is a \yes-instance of \ams, and let $\eta$ be one of its solution. Then, $|\mu \triangle \eta| \leq q$, and hence, $|V(\mu \triangle \eta)| \leq 2q$. Let $B = N_G(V(\mu \triangle \eta))$ and $D$ be the set of matching partners of the vertices in $B$, in the matching $\mu$. Since the maximum degree of a vertex in the graph $G$ is at most $d$, we have that $|B \cup D| \leq 2|B| \leq 4qd$. Since $\mathcal{F}_1$ is a $n$-$2q$-$4qd$-lopsided universal family, there exists a set $F \in \mathcal{F}_1$ such that $V(\mu \triangle \eta) \subseteq F$ and $(B\cup D)\cap F = \emptyset$. Let $f_F$ be the function corresponding to the set $F$. For $i\in [2]$, let $V_i$ be the set of colored $i$ vertices using the function $f_F$. Let $G_1 = G[V_1]$ and $G'=G_1[V(\mu \triangle \eta)]$. Since the maximum degree of a vertex in the graph $G$ is at most $d$ and  $|V(\mu \triangle \eta)| \leq 2q$, the number of edges in $G'$ is $2qd$. Since $\mathcal{F}_2$ is a $m$-$q$-$2qd$-lopsided universal family, there exists a set $F' \in \mathcal{F}_2$ such that $\mu \triangle \eta \subseteq F'$ and $(E(G')\setminus (\mu\triangle \eta))\cap F' = \emptyset$. Let $g_{F'}$ be the function corresponding to the set $F'$. Let $G^\star$ be the graph as constructed above in the randomized algorithm corresponding to the functions $f_{F'}$ and $g_{F'}$. Clearly, $\eta$ satisfies properties in the statement of Lemma~\ref{label:2coloredgesinsolA}. Thus, using the same arguments as in the proof of Lemma~\ref{lem:phase3}, we obtained that the algorithm returns ``yes''. For the other direction of the proof, if for any pair of $(f_F,g_{F'})$, where $F\in \mathcal{F}_1$ and $F'\in \mathcal{F'}$, the constructed instance of \wtkp is a \yes-instance of the problem, then as argued in the proof of Lemma~\ref{lem:kp-to-ams-proof}, $(G,\Co{L},\mu,k,q,t)$ is a \yes-instance of the problem. This completes the correctness of the algorithm.
\par
Note that the running time of the algorithm is upper bounded by $ |\mathcal{F}_1|\times | \mathcal{F}_2 |n^{\OO(1)}$. This results in the running time of the form $2^{\OO(q \log d)+o(dq)}n^{\OO(1)}$. To bound the running time we use the well known combinatorial identity that ${n \choose k} \leq (\frac{en}{k})^k$, concluding the proof of Theorem~\ref{th:fptAlgo}.  \qed

%

\section{Conclusion}
In this paper, we initiated the study of the computational complexity of the tradeoff between size and stability through the lenses of both local search and multivariate analysis. We wish to mention that the hardness results of Theorems~\ref{thm:asm}--\ref{thm:whard-k-q-t} hold even in the highly restrictive setting where {\it every preference list respects a master list}, \ie the relative ordering of the vertices in a preference list is same as that in a {\it master list}, a fixed ordering of all the vertices on the other side.  
This setting ensures that even when the preference lists on either side are both {\it single peaked} and {\it single-crossing} our hardness results hold true.
%
We conclude the paper with a few directions for further research.
\begin{itemize}[wide=0pt] 

\item In certain scenarios,  the ``satisfaction'' of the agents (there exist several measures such as {\it egalitarian}, {\it sex-equal}, {\it balance}) might be of importance. Then, it might be of interest to study the tradeoff between $t$ and $k$ while being $q$-away from the egalitarian stable matching.  

\item The formulation of \ams\ can be generalized to the {\sc Stable Roommates} problem (where graph $G$ may not be bipartite), or where the input contains a utility function on the edges and the objective is to maximize the value of a solution matching subject to this function. 
\item Lastly, we believe that the examination of the tradeoff between size and stability in real-world instances is of importance as it may shed light on the values of $k$ and $q$ that, in a sense, lead to the ``best'' exploitation of the tradeoff in practice.
\end{itemize}
\bibliographystyle{splncs04}
\bibliography{Voting.bib}

\end{document}